\documentclass[10pt, journal]{IEEEtran}

\usepackage{graphicx}
\usepackage{url}
\usepackage{color}
\usepackage{multirow}
\usepackage{amsmath}
\usepackage{amssymb}
\usepackage{algorithm}
\usepackage{algorithmicx}
\usepackage{algpseudocode}
\usepackage{cite}
\usepackage{amsthm}
\usepackage{bm}
\usepackage{xspace}
\usepackage{enumitem}
\usepackage{microtype}
\usepackage{amsmath}
\usepackage{pbox}

\usepackage{mathptmx}
\DeclareMathAlphabet{\mathcal}{OMS}{cmsy}{m}{n}
\DeclareMathAlphabet{\mathrm}{OT1}{bch}{m}{n}
\DeclareMathAlphabet{\mathit}{OT1}{bch}{m}{it}

\DeclareMathOperator*{\argmax}{arg\,max}
\newtheorem{theorem}{Theorem}

\newtheorem{lemma}{Lemma}

\newcommand{\sysname}{MV-Sketch\xspace}
\newcommand{\CS}{\mathcal{S}}
\newcommand{\CD}{\mathcal{D}}
\newcommand{\para}[1]{\noindent\textbf{#1}}

\makeatletter
\newcommand\fs@betterruled{%
  \def\@fs@cfont{\bfseries}\let\@fs@capt\floatc@ruled
  \def\@fs@pre{\vspace*{0pt}\hrule height.8pt depth0pt \kern2pt}%
  \def\@fs@post{\kern2pt\hrule\relax}%
  \def\@fs@mid{\kern2pt\hrule\kern2pt}%
  \let\@fs@iftopcapt\iftrue}
\floatstyle{betterruled}
\restylefloat{algorithm}
\makeatother

\begin{document}

%----------------------------------------------------------------------------
\title{A Fast and Compact Invertible Sketch for Network-Wide Heavy Flow
Detection}

\author{Lu Tang,
        Qun Huang,
        Patrick P. C. Lee%
\thanks{This work was supported in part by Research Grants Council of Hong
Kong (GRF 14204017), National Key R\&D Program of China (2019YFB1802600), and
National Natural Science Foundation of China (Grant No. 61802365).}
\thanks{An earlier version of this paper appeared at \cite{Tang2019}. In this
extended version, we extend \sysname for network-wide heavy flow detection and
present its implementation in programmable switches.  We also add new
evaluation results in both software and hardware environments.}
\thanks{L. Tang and P. P. C. Lee are with the Department of Computer Science
and Engineering, The Chinese University of Hong Kong, Hong Kong, China
(Emails:\{ltang,pclee\}@cse.cuhk.edu.hk).}% 
\thanks{Q. Huang is with Department of Computer Science and Technology, Peking
University, Beijing, China (Email: huangqun@pku.edu.cn).}
\thanks{Corresponding author: Patrick P. C. Lee.}
}

% The paper headers
%\markboth{
%   IEEE/ACM TRANSACTIONS ON NETWORKING}
%{Tang \MakeLowercase{\textit{et al.}}: A Fast and Compact Invertible Sketch
%for Network-Wide Heavy Flow Detection}

\maketitle

%----------------------------------------------------------------------------
% Abstract
%----------------------------------------------------------------------------
\begin{abstract} 
Fast detection of heavy flows (e.g., heavy hitters and heavy changers) in
massive network traffic is challenging due to the stringent requirements of
fast packet processing and limited resource availability.  Invertible sketches
are summary data structures that can recover heavy flows with small memory
footprints and bounded errors, yet existing invertible sketches incur high
memory access overhead that leads to performance degradation.  We present
\sysname, a fast and compact invertible sketch that supports heavy flow
detection with small and static memory allocation.  \sysname tracks candidate
heavy flows inside the sketch data structure via the idea of majority voting,
such that it incurs small memory access overhead in both update and query
operations, while achieving high detection accuracy.  We present theoretical
analysis on the memory usage, performance, and accuracy of \sysname in both
local and network-wide scenarios.  We further show how \sysname can be
implemented and deployed on P4-based programmable switches subject to hardware 
deployment constraints. We conduct evaluation in both software and hardware
environments.  Trace-driven evaluation in software shows that \sysname
achieves higher accuracy than existing invertible sketches, with up to
3.38$\times$ throughput gain. We also show how to boost the performance of
\sysname with SIMD instructions.  Furthermore, we evaluate \sysname on a
Barefoot Tofino switch and show how \sysname achieves line-rate measurement
with limited hardware resource overhead. 
\end{abstract}


%---------------------------------------------------------------------------
% Introduction
%---------------------------------------------------------------------------
\section{Introduction}
\label{sec:introduction}

\IEEEPARstart{I}{dentifying} abnormal patterns of {\em flows} (e.g., hosts,
source-destination pairs, or 5-tuples) in massive network traffic is essential
for various network management tasks, such as traffic engineering
\cite{Benson2011}, load balancing \cite{Alizadeh2014} and intrusion detection
\cite{Garcia2009}.  
Two types of abnormal flows are of particular interest: {\em heavy hitters}
(i.e., flows that generate an unexpectedly high volume of traffic) and 
{\em heavy changers} (i.e., flows that generate an unexpectedly high change of
traffic volume in a short duration).  By identifying both heavy hitters and
heavy changers (collectively referred to as {\em heavy flows}), network
operators can quickly respond to performance outliers, mis-behaved usage, and
potential DDoS attacks, so as to maintain network stability and QoS
guarantees.  

Unfortunately, the stringent requirements of fast packet processing and
limited memory availability pose challenges to practical heavy flow detection.
First, the packet processing rate of heavy flow detection must keep pace with
the ever-increasing network speed, especially in the worst case when traffic
bursts or attacks happen \cite{Huang2017}.  For example, a fully utilized
10\,Gb/s link with a minimum packet size of 64~bytes implies that the heavy
flow detection algorithm must process at least 14.88M packets per second.  In
addition, the available memory footprints are constrained in practice 
(e.g., less than 2\,MB of SRAM per stage in emerging programmable switches 
\cite{Bosshart2013,Sivaraman2017}). 
While per-flow monitoring with linear hash tables is arguably feasible in
software \cite{Alipourfard2018}, its performance degrades once
the working set grows beyond the available software cache capacity.

Given the rigid packet processing and memory requirements, many approaches
perform approximate heavy flow detection via {\em sketches}, which are summary
data structures that significantly mitigate memory footprints with bounded
detection errors.  Classical sketches
\cite{Cormode2005,Charikar2002,Krishnamurthy2003} are proven effective, but
are {\em non-invertible}: while we can query a sketch whether a specific
flow is a heavy flow, we cannot readily recover all heavy flows from only the
sketch itself.  Instead, we must check whether every possible flow is a heavy
flow. Such a brute-force approach is computationally expensive for an
extremely large flow key space (e.g., the size is $2^{104}$ for 5-tuple
flows). 

This motivates us to explore {\em invertible sketches}, which provide provable
error bounds as in classical sketches, while supporting the queries of
recovering all heavy flows.  
Invertible sketches are well studied in the literature (e.g.,
\cite{Cormode2005,Cormode2005deltoid,Schweller2007,Bu2010,Liu2012,Huang2014})
for heavy flow detection.  However, there remain limitations in existing
invertible sketches.  In particular, they either maintain heavy flows in
external DRAM-based data structures \cite{Cormode2005,Huang2014}, or track
flow keys in smaller-size bits or sub-keys
\cite{Cormode2005deltoid,Schweller2007,Bu2010,Liu2012}. We argue that both
approaches incur substantial memory access overhead that leads to degraded
processing performance (Section~\ref{subsec:sketches}).

In this paper, we present \sysname, a fast and compact invertible sketch for
heavy flow detection.  It tracks candidate heavy flow keys together with the
counters in a sketch data structure, and updates the candidate heavy flow keys
based on the {\em majority vote algorithm} \cite{Boyer1991} in an online
streaming fashion.  A key design feature of \sysname is that it maintains a
sketch data structure with {\em small} and {\em static} memory allocation
(i.e., no dynamic memory allocation is needed). This not only allows
lightweight memory access in both update and detection operations, but also
provides viable opportunities for hardware acceleration and feasible
deployment in hardware switches.  To summarize, we make the following
contributions. 
\begin{itemize}[leftmargin=*]
\item 
We design \sysname, an invertible sketch that supports both heavy hitter and
heavy changer detection and can be generalized for distributed detection,
including both scalable detection (which provides scalability) and
network-wide detection (which provides a network-wide view of measurement
results). See Section~\ref{sec:overview}. 
\item 
We present theoretical analysis on \sysname for its memory space
complexity, update/detection time complexity, and detection accuracy. See 
Section~\ref{sec:theory}. 
\item 
We present the implementation of \sysname on P4-based programmable switches
\cite{Bosshart2014} subject to various hardware deployment constraints. See
Section~\ref{sec:tofino}.
\item 
We conduct evaluation in both software and hardware environments. We show via
trace-driven evaluation in software that \sysname achieves higher detection
accuracy for most memory configurations and up to 3.38$\times$ throughput gain
over state-of-the-art invertible sketches.  We also extend \sysname
with {\em Single Instruction, Multiple Data (SIMD)} instructions to boost its
update performance.  Furthermore, we prototype \sysname in the P4 language
\cite{Bosshart2014} and compile it to the Barefoot Tofino chipset
\cite{tofino}.  \sysname achieves line-rate measurement with limited resource
overhead.  It also achieves higher accuracy and smaller resource usage than
PRECISION \cite{Ben2018precision} (a heavy hitter detection scheme in
programmable switches).  See Section~\ref{sec:evaluation}. 
\end{itemize}

The source code of \sysname (including the software implementation and the P4
code) is available for download at: 
{\bf http://adslab.cse.cuhk.edu.hk/software/mvsketch}. 

%------------------------------------------------------------------------
% Background
%------------------------------------------------------------------------
\section{Background} 
\label{sec:background}

\subsection{Heavy Flow Detection} 
\label{subsec:heavy_flow_basics}

We consider a stream of packets, each of which is denoted by a key-value pair
$(x, v_x)$, where $x$ is a key drawn from a domain $[n] = \{0, 1, \cdots,
n-1\}$ and $v_x$ is the value of $x$.  In network measurement, $x$ is the flow
identifier (e.g., source/destination address pairs or 5-tuples), while $v_x$
is either one (for packet counting) or the packet size (for byte counting).
We conduct measurement at regular time intervals called {\em epochs}.  

We formally define heavy hitters and heavy changers as follows.  Let $\phi$ be
a pre-defined fractional threshold (where $0<\phi<1$) that is used to
differentiate heavy flows from network traffic (we use the same $\phi$ for
both heavy hitter and heavy changer detection for simplicity).  Let $S(x)$ be
the sum (of all $v_x$'s) of flow $x$ in an epoch, and $D(x)$ be the absolute
change of $S(x)$ of flow $x$ across two epochs.  Let $\CS$ be the total sum of
all flows in an epoch (i.e., $\CS = \sum_{x\in[n]}S(x)$), and $\CD$ be the total
absolute change of all flows across two epochs (i.e., $\CD =
\sum_{x\in[n]}D(x)$).  Both $\CS$ and $\CD$ can be obtained in practice: for
$\CS$, we can maintain an extra counter that counts the total traffic; for
$\CD$, we can run an $l_1$-streaming algorithm and estimate $\CD$ (equivalent
to the $l_1$-distance) in one pass \cite{Nelson2010}.  Finally, flow $x$ is
said to be a {\em heavy hitter} if $S(x) \ge \phi \CS$, or a {\em heavy
changer} if $D(x) \ge \phi \CD$. 

\subsection{Sketches}
\label{subsec:sketches}

{\em Sketches} are summary data structures that track values in a fixed number
of entries called {\em buckets}.  Classical sketches on heavy flow
detection (e.g., Count Sketch \cite{Charikar2002}, K-ary Sketch
\cite{Krishnamurthy2003}, and Count-Min Sketch \cite{Cormode2005}) represent a
sketch as a two-dimensional array of buckets and provide different
theoretical trade-offs across memory usage, performance, and accuracy.

Take Count-Min Sketch \cite{Cormode2005} as an example.  We construct the
sketch as $r$ rows of $w$ buckets each.  Each bucket is
associated with a counter initialized as zero. For each tuple $(x,v_x$)
received in an epoch, we hash $x$ into a bucket in each of the $r$ rows using
$r$ pairwise independent hash functions.  We increment the counter in each of
the $r$ hashed buckets by $v_x$.  Since multiple flows can be hashed to the
same bucket, we can only provide an estimate for the sum of a flow.  Count-Min
Sketch uses the minimum counter value of all $r$ hashed buckets as the
estimated sum of a flow.  We can check if a flow is a heavy hitter by checking
if its estimated sum exceeds the threshold; similarly, we can check if a flow
is a heavy changer by checking if the absolute change of its estimated sums in
two epochs exceeds the threshold.  However, Count-Min Sketch is
non-invertible, as we must check every flow in the entire flow key space to
recover all heavy flows; note that Count Sketch and K-ary Sketch are also
non-invertible. 

Invertible sketches (e.g.,
\cite{Cormode2005deltoid,Cormode2005,Bu2010,Schweller2007,Huang2014,Liu2012}) 
allow all heavy flows to be recovered from only the sketch data structure
itself.  State-of-the-art invertible sketches can be classified into three
types.

\para{Extra data structures.} Count-Min-Heap \cite{Cormode2005} is an
augmented Count-Min Sketch that uses a heap to track all candidate heavy flows
and their estimated sums.  If any incoming flow whose estimated sum
exceeds the threshold, it is added to the heap.  LD-Sketch \cite{Huang2014}
maintains a two-dimensional array of buckets and links each bucket with
an associative array to track the candidate heavy flows that are hashed to the
bucket.  However, updating a heap or an associative array incurs high memory
access overhead, which increases with the number of heavy flows.  In
particular, LD-Sketch occasionally expands the associative array to hold more
candidate heavy flows, yet dynamic memory allocation is a costly operation and
difficult to implement in hardware \cite{Basat2016}. 

\para{Group testing.} Deltoid \cite{Cormode2005deltoid} 
comprises multiple counter groups with $1+L$ counters each (where
$L$ is the number of bits in a key), in which one counter tracks the total
sum of the group, and the remaining $L$ counters correspond to the bit
positions of a key.  It maps each flow key to a subset of groups and
increments the counters whose corresponding bits of the key are one.  To
recover heavy flows, Deltoid first identifies all groups whose total sums 
exceed the threshold.  If each such group has only one heavy flow, the heavy
flow can be recovered: the bit is one if a counter exceeds the threshold, or
zero otherwise.  Fast Sketch \cite{Liu2012} is similar to Deltoid except that
it maps the quotient of a flow key to the sketch.  However, both Deltoid and
Fast Sketch have high update overhead, as their numbers of counters increase
with the key length.

\para{Enumeration.} Reversible Sketch \cite{Schweller2007} finds heavy
flows by pruning the enumeration space of flow keys.  It divides a flow key
into smaller sub-keys that are hashed independently, and concatenates
the hash results to identify the hashed buckets.  To recover heavy flows, it
enumerates each sub-key space and combines the recovered sub-keys to form the
heavy flows.  SeqHash \cite{Bu2010} follows a similar design, yet it hashes
the key prefixes of different lengths into multiple smaller sketches.
However, the update costs of both Reversible Sketch and SeqHash increase with
the key length. 

%---------------------------------------------------------------------------
% MV-Sketch Design
%---------------------------------------------------------------------------
\section{\sysname Design}
\label{sec:overview}

\sysname is a novel invertible sketch for heavy flow detection and 
aims for the following design goals: 
\begin{itemize}[leftmargin=*]
\item
{\bf Invertibility:} \sysname is invertible and readily returns all heavy
flows (i.e., heavy hitters or heavy changers) from only the sketch data
structure itself. 
\item
{\bf High detection accuracy:} \sysname supports accurate heavy flow detection
with provable error bounds. 
\item
{\bf Small and static memory:} \sysname maintains compact data structures with
small memory footprints.  Also, it can be constructed with static memory
allocation, which mitigates memory management overhead as opposed to dynamic
memory allocation \cite{Huang2014}. 
\item
{\bf High processing speed:} \sysname processes packets at high speed by
limiting the memory access overhead of per-packet updates.  It also takes
advantage of static memory allocation to allow hardware acceleration. 
\item
{\bf Scalable detection:}  To improve performance and scalability, \sysname
can be extended for scalable detection by processing packets in multiple
\sysname instances in parallel. 
\item  
{\bf Network-wide detection:}  \sysname can provide a network-wide view of
heavy flows by aggregating the results from multiple \sysname instances
deployed in different measurement points across the whole network. 
\end{itemize}

\subsection{Main Idea}
\label{subsec:idea}

Like Count-Min Sketch \cite{Cormode2005}, \sysname is initialized as a
two-dimensional array of buckets (Section~\ref{subsec:sketches}), in
which each bucket tracks the values of the flows that are hashed to itself.
\sysname augments Count-Min Sketch to allow each bucket to also track a {\em
candidate heavy flow} that has a high likelihood of carrying the largest
amount of traffic among all flows that are hashed to the bucket.  Our
rationale is that in practice, a small number of large flows dominate in IP
traffic \cite{Zhang2002}.  Thus, the candidate heavy flow is very likely to
carry much more traffic than all other flows that are hashed to the same
bucket.  Also, by hashing each flow to multiple buckets independently, we
can significantly reduce the error due to the collision of heavy flows, so as
to accurately track multiple heavy flows.

To find the candidate heavy flow in each bucket, we apply the {\em majority
vote algorithm (MJRTY)} \cite{Boyer1991}, which enables us to track the
candidate heavy flow in an online streaming fashion.  MJRTY processes a stream
of votes (corresponding to packets in our case), each of which has a vote key
and a vote count one.  It aims to find the {\em majority vote}, defined as the
vote key that has more than half of the total vote counts, from the stream of
votes in one pass with constant memory usage. At any time, it stores (i) the
{\em candidate majority vote} that is thus far observed in a stream and (ii)
an {\em indicator counter} that tracks whether the currently stored vote
remains the candidate majority vote. Initially, it stores the first vote and
initializes the indicator counter as one.  Each time when a new vote arrives,
MJRTY compares the new vote with the candidate majority vote. If both votes
are the same (i.e., the same vote key), it increments the indicator counter by
one; otherwise, it decrements the indicator counter by one.  If the indicator
counter is below zero, MJRTY replaces the current candidate majority vote with
the new vote and resets the counter to one.  MJRTY ensures that the true
majority vote must be the candidate majority vote stored by MJRTY at the end
of the stream \cite{Boyer1991}. 

\sysname addresses the limitations of existing solutions as follows. First, it
supports static memory allocation and does not maintain any complex data
structure, thereby avoiding the high memory access overhead in Count-Min-Heap
\cite{Cormode2005} and LD-Sketch \cite{Huang2014}.  Also, \sysname stores
candidate heavy flows in buckets via MJRTY, thereby reducing the update
overhead of sketches that are based on group testing
\cite{Cormode2005deltoid,Liu2012} and enumeration \cite{Schweller2007,Bu2010}.

\begin{figure}[!t]
\centering
\includegraphics[width=3.4in]{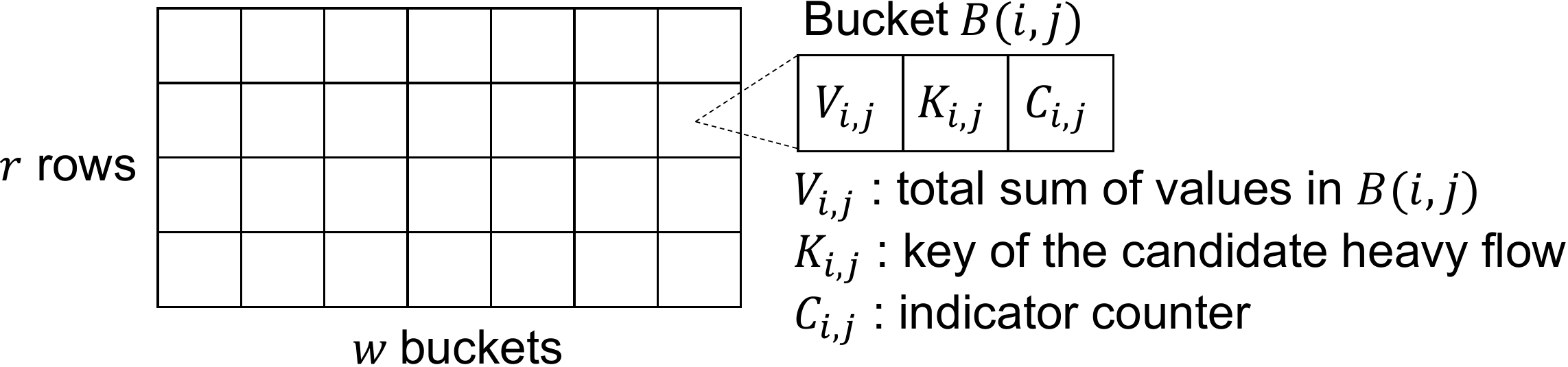}
\vspace{-6pt}
\caption{Data structure of \sysname.}
\label{fig:msketch}
\vspace{-6pt}
\end{figure}

\subsection{Data Structure of \sysname}
\label{subsec:data_structure}

Figure~\ref{fig:msketch} shows the data structure of \sysname, which is
composed of a two-dimensional array of buckets with $r$ rows and $w$ columns.
Let $B(i,j)$ denote the bucket at
the $i$-th row and the $j$-th column, where $1\le i\le r$ and $1\le j\le w$.
Each bucket $B(i,j)$ consists of three fields: (i) $V_{i,j}$, which counts the
total sum of values of all flows hashed to the bucket; (ii) $K_{i,j}$, which
tracks the key of the current candidate heavy flow in the bucket; and (iii)
$C_{i,j}$, which is the indicator counter that checks if the candidate heavy
flow should be kept or replaced as in MJRTY \cite{Boyer1991}.  In addition,
\sysname is associated with $r$ pairwise-independent hash functions, denoted
by $h_1 \ldots h_r$, such that each $h_i$ (where $1\le i\le r$) hashes the key
$x\in[n]$ of each incoming packet to one of the $w$ buckets in row $i$.  Note
that the data structure has a fixed memory size and can be pre-allocated in
advance. 

\subsection{Basic Operations}
\label{subsec:basic}

\sysname supports two basic operations: (i) {\em Update}, which inserts each
incoming packet into the sketch; (ii) {\em Query}, which returns the estimated
sum of a given flow in an epoch.  

Algorithm~\ref{alg:update} shows the Update operation. All fields 
$V_{i,j}$, $K_{i,j}$, and $C_{i,j}$ are initialized as zero for $B(i,j)$,
where $1\le i\le r$ and $1\le j\le w$.  For each $(x, v_x)$, we hash $x$ into
$B(i,j)$ in the $i$-th row with $j = h_i(x)$ for $1\le i\le r$. We first
increment $V_{i,j}$ by $v_{x}$ (Line~2). We then check if $x$ is stored in
$K_{i,j}$ based on the MJRTY algorithm: if $K_{i,j} $ equals $ x$, we increment
$C_{i,j}$ by $v_x$ (Lines~3-4). Otherwise, we decrement $C_{i,j}$ by $v_x$
(Lines~5-6); if $C_{i,j}$ drops below zero, we replace $K_{i,j}$ by $x$ and
reset $C_{i,j}$ with its absolute value (Lines~7-10).  Note that the Update
operation differs from MJRTY as it supports general value counts (or the
number of bytes) with any non-negative value $v_x$, while MJRTY considers only
vote counts (or the number of packets) with $v_x$ always being one. 

\begin{algorithm}[t]
\caption{Update}
\label{alg:update}
\begin{small}
\begin{algorithmic}[1]
\item[\textbf{Input:}{$(x, v_x)$}]
\For {$i = 1$ to $r$}
  \State $V_{i,h_{i}(x)} \leftarrow V_{i,h_{i}(x)} + v_{x}$
  \If {$K_{i,h_{i}(x)} = x$} 
    \State $C_{i,h_{i}(x)} \leftarrow C_{i,h_{i}(x)} + v_{x}$
  \Else
    \State $C_{i,h_{i}(x)} \leftarrow C_{i,h_{i}(x)} - v_{x}$
    \If {$C_{i,h_{i}(x)} < 0$}
      \State $K_{i,h_{i}(x)} \leftarrow x$
      \State $C_{i,h_{i}(x)} \leftarrow -C_{i,h_{i}(x)}$
    \EndIf 
  \EndIf
\EndFor
\end{algorithmic}
\end{small}
\end{algorithm}

Algorithm~\ref{alg:query} shows the Query operation.  For each hashed bucket
in row $i$ (where $1\le i\le r$), we calculate a row estimate $\hat{S}_i(x)$
of flow $x$ (Lines~1-7): if $x$ and $K_{i,j}$ are the same, we set
$\hat{S}_i(x) = (V_{i,j}+C_{i,j})/2$; otherwise, we set $\hat{S}_i(x) =
(V_{i,j}-C_{i,j})/2$.  
The row estimate $\hat{S}_i(x)$ is actually the upper bound of $x$ in the
hashed bucket $B(i,j)$ (Lemma~\ref{lem:bucketbound}).
Finally, we return the final estimate $\hat{S}(x)$ as
the minimum of all row estimates (Lines~8-9). 

\begin{algorithm}[!t]
\caption{Query}
\label{alg:query}
\begin{small}    
\begin{algorithmic}[1]
\item[\textbf{Input:}{ flow key $x$}]
\item[\textbf{Output:}{ estimate $\hat{S}(x)$ of flow $x$}]
\For {$i = 1$ to $r$}
  \If {$K_{i, h_{i}(x)} = x $}      
    \State $\hat{S}_{i}(x) \leftarrow (V_{i,h_{i}(x)} + C_{i,h_{i}(x)})/2$
  \Else
    \State $\hat{S}_{i}(x) \leftarrow (V_{i,h_{i}(x)} - C_{i,h_{i}(x)})/2$
  \EndIf
\EndFor
\State $\hat{S}(x) \leftarrow \min_{1\le i\le r}\{\hat{S}_{i}(x)\}$
\State \Return $\hat{S}(x)$
\end{algorithmic}
\end{small}
\end{algorithm}

\subsection{Heavy Flow Detection}
\label{subsec:detection}

\para{Heavy hitter detection.}
To detect heavy hitters, we check every bucket $B(i,j)$ ($1\le i\le r$ and
$1\le j\le w$) at the end of an epoch.  For each $B(i,j)$, if $V_{i,j}\ge
\phi\CS$, we let $x = K_{i,j}$ and query $\hat{S}(x)$ via
Algorithm~\ref{alg:query}; if $\hat{S}(x)\ge \phi\CS$, we report $x$ as a
heavy hitter.  

\para{Heavy changer detection.}
To detect heavy changers, we compare two sketches at the ends of two epochs.
One possible detection approach is to exploit the linear property of sketches
as in prior studies \cite{Cormode2005deltoid,Liu2012,Schweller2007}, in which
we compute the differences of $V_{i,j}$'s of the buckets at the same positions
across the two sketches and recover the flows from the buckets whose
differences exceed the threshold $\phi \CD$
(Section~\ref{subsec:heavy_flow_basics}).  However, such an approach
can return many false negatives, since the hash collisions of two heavy
changers, one with a high incremental change and another with a high
decremental change, can cancel out the changes of each other. 

To reduce the number of false negatives, we instead use the {\em estimated 
maximum change} of a flow for heavy changer detection.  Specifically, let
$U(x)$ and $L(x)$ be the upper and lower bounds of $S(x)$, respectively.  We
set $U(x) = \hat{S}(x)$ returned by Algorithm~\ref{alg:query}.  Also,
we set $L(x) = \max_{1\le i\le r}\{L_i(x)\}$, where $L_i(x)$ is set as
follows: for each hashed bucket $B(i,j)$ of $x$ (where $1\le i\le r$ and
$j=h_i(x)$), if $K_{i,j}$ equals $ x$, we set $L_i(x) = C_{i,j}$; otherwise,
we set $L_i(x)=0$.  Note that both $U(x)$ and $L(x)$ are the {\em true} upper
and lower bounds of $S(x)$, respectively (Lemma~\ref{lem:bucketbound} in
Section~\ref{subsec:theory_hh}).  Now, let $U^1(x)$ and $L^1(x)$ (resp.
$U^2(x)$ and $L^2(x)$) be the upper and lower bounds of $S(x)$ in the previous
(resp. current) epoch, respectively.  Then the estimated maximum change of flow
$x$ is given by $\hat{D}(x)=\max\{|U^1(x)- L^2(x)|, |L^1(x)-U^2(x)|\}$.  

We now detect heavy changers as follows.  We check every bucket $B(i,j)$
($1\le i\le r$ and $1\le j\le w$) of two sketches of the previous and current
epochs.  For each $B(i,j)$ in each of the sketches, if $V_{i,j}\ge\phi\CD$, we
let $x = K_{i,j}$ and estimate $\hat{D}(x)$; if $\hat{D}(x)\ge\phi\CD$, we
report $x$ as a heavy changer (note that a necessary condition of
a heavy changer is that its sum must exceed the threshold in at least one
epoch).  

Currently, \sysname focuses on the values (e.g., packet or byte counts) of a
flow.  We can extend \sysname to monitor hosts with a high number of distinct
connections in DDoS or superspreader detection by either associating the
buckets with approximate distinct counters \cite{Cormode2005graph} or
filtering duplicate connections with a Bloom filter \cite{Zhao2005}.  

\subsection{Scalable Heavy Flow Detection}
\label{subsec:distributed}

We can improve the performance and scalability of \sysname by performing heavy
flow detection on multiple packet streams in parallel based on a distributed
streaming architecture \cite{Huang2014}.  Specifically, we deploy $q \ge 1$
detectors, each of which deploys an \sysname instance to monitor packets from
multiple streaming sources.  Suppose that each streaming source maps a flow to
a subset $d$ out of $q$ detectors, where $d\le q$, and dispatches each packet
of the flow {\em uniformly} to one of the $d$ selected detectors.  At the end
of each epoch, each detector sends the local detection results to a
centralized controller for final heavy flow detection. 

For heavy hitter detection, each detector checks every bucket $B(i,j)$ in
\sysname.  Let $x = K_{i,j}$, and if $\hat{S}(x) \ge \tfrac{\phi}{d}\CS$, the
detector sends the tuple ($x$, $\hat{S}(x)$) of flow $x$ to the controller.
After collecting all results from $q$ detectors, the controller adds the
estimates of each flow.  If the added estimate of a flow exceeds $\phi\CS$,
the flow is reported as a heavy hitter. 

For heavy changer detection, each detector checks every bucket $B(i,j)$ of two
sketches of the previous and current epochs.  If
$V_{i,j}\ge\tfrac{\phi}{d}\CD$, it lets $x = K_{i,j}$ and  estimates
$\hat{D}(x)$; if $\hat{D}(x) \ge\tfrac{\phi}{d}\CD$, the detector sends the
tuple ($x$, $\hat{D}(x)$) of flow $x$ to the controller.  
The controller adds the estimates of each flow from $q$ detectors.
If the added estimate of a flow exceeds $\phi\CD$, the flow is reported as a
heavy changer. 

\subsection{Network-Wide Heavy Flow Detection}
\label{subsec:networkwide}

\begin{algorithm}[t]
\caption{Merge}
\label{alg:merge}
\begin{small}
\begin{algorithmic}[1]
\item[\textbf{Input:}{ $q$ \sysname instances}]
\item[\textbf{Output:}{ the merged \sysname}]
\For {$i = 1$ to $r$}
  \For {$j = 1$ to $w$}
    \State $V_{i,j} \leftarrow \sum_{1 \le k\le q}V^{k}_{i,j}$ 
    \State $ T \leftarrow \{K^{1}_{i,j}, K^{2}_{i,j}, \dots , K^{q}_{i,j}\}$
    \For {$x \in T$}
      \State $e(x) \leftarrow 0$ 
      \For {$k = 1$ to $q$}
        \If {$K^{k}_{i,j} = x$}
            \State $e(x) \leftarrow  e(x) + (V^k_{i,j}+C^k_{i,j})/2 $
        \Else
            \State $e(x) \leftarrow  e(x) + (V^k_{i,j}-C^k_{i,j})/2 $
        \EndIf 
      \EndFor
    \EndFor
    \State $K_{i,j} \leftarrow x^*$, where $x^* = \argmax_{x\in T} \{e(x)\}$
    \State $C_{i,j} \leftarrow \max\{2e(x^*) - V_{i,j}, 0\}$
  \EndFor
\EndFor
\end{algorithmic}
\end{small}
\end{algorithm}

We can also perform network-wide heavy-flow detection via \sysname by
deploying multiple \sysname instances in multiple detectors (e.g., end-hosts
or programmable switches) that span across the whole network and aggregating
the measurement results from all detectors in a centralized controller, as in
recent sketch-based network-wide measurement systems
\cite{Yu2013,Moshref2015,Liu2016,Li2016flowradar,Huang2017,Huang2018,Yang2018}.
While scalable detection (Section~\ref{subsec:distributed}) focuses on
improving scalability by processing multiple packet streams in parallel,
network-wide detection aims to provide an accurate network-wide measurement
view as if all traffic were measured in one big detector \cite{Liu2016}. 

In network-wide heavy flow detection, we deploy an \sysname instance in each
of the detectors, such that all \sysname instances share the same hash
functions and parameter settings.  We assume that each packet being monitored
appears at only one \sysname instance to avoid duplicate measurement (e.g., by
monitoring only the ingress or egress traffic).  We deploy a centralized
controller that collects and merges the \sysname instances from all detectors.
Note that we do not make any assumption on the traffic size distribution of a
heavy flow in each detector (e.g., a heavy flow may have small traffic size in
some detectors); this is in contrast to scalable detection
(Section~\ref{subsec:distributed}), in which we assume that the traffic size
of a flow is uniformly distributed across detectors. 

Algorithm~\ref{alg:merge} shows how the controller merges multiple \sysname
instances.  Suppose that there are $q\ge 1$ detectors. Let $B^k(i,j)$ denote
the bucket in the $i$-th row and $j$-th column of \sysname in the $k$-th
detector, where $1\le i\le r$, $1\le j \le w$, and $1\le k\le
q$.  Let $V^k_{i,j}$, $K^k_{i,j}$, and $C^k_{i,j}$ denote the total sum of
values hashed to the bucket, the candidate heavy flow key, and the indicator
counter of $B^k(i,j)$, respectively.  Also, let $B(i,j)$ (with the
corresponding fields $V_{i,j}$, $K_{i,j}$, and $C_{i,j}$) be the bucket in the
$i$-th row and $j$-th column in the merged sketch.  In
Algorithm~\ref{alg:merge}, the controller constructs each bucket $B(i,j)$ of
the merged sketch by merging all $B^k(i,j)$ that have the same $i$ and $j$ in
all $q$ detectors.  The controller first sets $V_{i,j}$ as the sum of
$V^k_{i,j}$'s of all $q$ detectors (Line~3). It then calculates a network-wide
estimate $e(x)$ for each candidate heavy flow key $x\in T =
\{K^k_{i,j}\}_{1\le k\le q}$ (note that $x$ is hashed to every bucket
$B^k(i,j)$ for all $k$, where $1\le k\le q$).  First, the controller
initializes $e(x)$ as zero.  For each $B^k(i,j)$ ($1\le k \le q$), if
$K^k_{i,j}$ equals $x$, the controller increments $e(x)$ by
$(V^k_{i,j}+C^k_{i,j})/2$; otherwise, it increments $e(x)$ by
$(V^k_{i,j}-C^k_{i,j})/2$ (Lines~4-14). After that, the controller stores the
key $x^*$ that has the maximum estimate among all candidate heavy flow keys
into $K_{i,j}$ (Line~15).  It also sets $C_{i,j}$ as the maximum value of
$2e(x^*)-V_{i,j}$ and zero (Line~16).  By Lemma~\ref{lem:bucketbound} in
Section~\ref{subsec:theory_hh}, we can show that the network-wide estimate
$e(x)$ after Line~14 is an upper-bound of $S(x)$ (i.e., $e(x) \ge S(x)$). 

Once the controller finishes the merge operation, it performs heavy flow
detection on the merged \sysname as in Section~\ref{subsec:detection}. We show
that the merged \sysname achieves the same theoretical guarantee on accuracy
as in a single \sysname (Section~\ref{subsec:nwbound}).

%--------------------------------------------------------------------------
% Theoretical Analysis
%--------------------------------------------------------------------------
\section{Theoretical Analysis}
\label{sec:theory}

We present theoretical analysis on \sysname in heavy flow detection. We also
compare \sysname with several state-of-the-art invertible sketches. 

\subsection{Space and Time Complexities}

Our analysis assumes that \sysname is configured with $r =
\log{\tfrac{1}{\delta}}$ and $w = \tfrac{2}{\epsilon}$, where $\epsilon$
($0<\epsilon< 1$) is the approximation parameter, $\delta$ ($0<\delta<1$) is
the error probability, and the logarithm base is 2.
Theorem~\ref{the:complexity} states the space and time complexities of
\sysname. 

\begin{theorem}
\label{the:complexity}
The space usage is $O(\tfrac{1}{\epsilon}\log{\tfrac{1}{\delta}}\log{n})$. The
update time per packet is $O(\log\tfrac{1}{\delta})$, while the detection time
of returning all heavy flows is $O(\tfrac{1}{\epsilon}\log^2\tfrac{1}{\delta})$.
\end{theorem}

\begin{proof}
Each bucket of \sysname stores a $\log{n}$-bit candidate heavy flow and two
counters, so the space usage of \sysname is $O(rw\log n) =
O(\tfrac{1}{\epsilon}\log{\tfrac{1}{\delta}}\log{n})$.  

Each per-packet update accesses $r$ buckets and requires $r =
\log\tfrac{1}{\delta}$ hash operations, thereby taking
$O(\log\tfrac{1}{\delta})$ time. 

Returning all heavy flows requires to traverse all $rw$ buckets. For each
bucket whose $V_{i,j}$ is above the threshold, we check $r$ buckets to obtain
the estimate (either $\hat{S}(x)$ or $\hat{D}(x)$) for $x = K_{i,j}$.  This
takes $O(r^2w) = O(\tfrac{1}{\epsilon}\log^2\tfrac{1}{\delta})$ time.  
\end{proof}

\subsection{Error Bounds for Heavy Hitter Detection}
\label{subsec:theory_hh}

Suppose that for all flows hashed to
a bucket $B(i,j)$, flow $x$ is said to be a {\em majority flow} of $B(i,j)$ if
its sum $S(x)$ is more than half of the total value count $V_{i,j}$.  Then
Lemma~\ref{lem:majority} states that the majority flow must be tracked; note
that it is a generalization of the main result of MJRTY \cite{Boyer1991}. 

\begin{lemma} 
\label{lem:majority}
If there exists a majority flow $x$ in $B(i,j)$, then it must be stored in
$K_{i,j}$ at the end of an epoch. 
\end{lemma}
\begin{proof}
We prove by contradiction.  By definition, the majority flow $x$ has
$S(x)>\tfrac{1}{2}V_{i,j}$.  Suppose that $K_{i,j}\ne x$.  Then the
increments (resp. decrements) of $C_{i,j}$ due to $x$ must be offset by the
decrements (resp. increments) of other flows that are also hashed to $B(i,j)$.
This requires that $V_{i,j} - S(x) \ge S(x)$ (i.e., the total value count of
other flows is larger than $S(x)$).  Thus, $V_{i,j} \ge 2S(x) > V_{i,j}$,
which is a contradiction. 
\end{proof}

Lemma~\ref{lem:bucketbound} next bounds the sum $S(x)$ of flow $x$. 

\begin{lemma}
\label{lem:bucketbound}
Consider a bucket $B(i,j)$ that flow $x$ is hashed to. If $K_{i,j} $ equals $ x$, then 
$C_{i,j}\le S(x)\le \tfrac{V_{i,j}+C_{i,j}}{2}$; otherwise, $0\le S(x) \le
\tfrac{V_{i,j}-C_{i,j}}{2}$.
\end{lemma}

\begin{proof}
Suppose that $K_{i,j} $ equals $ x$. Let $\Delta$ be the offset amount of $x$ from
$C_{i,j}$ due to other flows. Then we have $S(x)= C_{i,j} + \Delta \ge
C_{i,j}$.  Also, since $V_{i,j} \ge S(x)+\Delta = C_{i,j}+2\Delta$, we have
$\Delta\le\tfrac{V_{i,j}-C_{i,j}}{2}$.  Thus, $S(x)=C_{i,j}+\Delta\le
\tfrac{V_{i,j}+C_{i,j}}{2}$.
   
Suppose now that $K_{i,j}\ne x$.  Then the increments (resp. decrements) of
$C_{i,j}$ due to $x$ must be offset by the decrements (resp. increments) made
by other flows that are also hashed to the same bucket (see the proof of
Lemma~\ref{lem:majority}).  The total value count of all flows other than $x$
(i.e., $V_{i,j}-S(x)$)  minus the offset amount $S(x)$ is at least $C_{i,j}$.
Thus, we have $V_{i,j} \ge C_{i,j}+2S(x)$, implying that $0\le
S(x)\le\tfrac{V_{i,j}-C_{i,j}}{2}$.
\end{proof}

We now study the bounds of the estimated sum $\hat{S}(x)$ of flow $x$ returned
by Algorithm~\ref{alg:query}.  From Lemma~\ref{lem:bucketbound} and the
definition of $\hat{S}(x)$ in Algorithm~\ref{alg:query}, we see that
$\hat{S}(x) \ge S(x)$.  Also, Lemma~\ref{lem:pointquery} states the upper
bound of $\hat{S}(x)$ in terms of $\epsilon$ and $\delta$. 
\begin{lemma}
\label{lem:pointquery}
$\hat{S}(x) \le S(x)+\tfrac{\epsilon \CS}{2}$ with a probability at least
$1-\delta$.
\end{lemma}
\begin{proof}
Consider the expectation of the total sum of all flows except $x$ in each
bucket $B(i,j)$.  It is given by $E[V_{i,j}-S(x)] = 
E[\textstyle\sum_{y\ne x, h_i(y) = h_i(x)}S(y)] \le \tfrac{\CS-S(x)}{w} \le
\tfrac{\epsilon \CS}{2}$ due to the pairwise independence of $h_i$ and the
linearity of expectation. By Markov's inequality, we have
\begin{equation}
\label{eq:markov}
\Pr[V_{i,j}-S(x)\ge \epsilon \CS]\le \tfrac{1}{2}.
\end{equation}

We now consider the row estimate $\hat{S}_i(x)$ (Algorithm~\ref{alg:query}).
If $K_{i,j}$ equals $ x$, then $\hat{S}_i(x)-S(x) =
\tfrac{V_{i,j}+C_{i,j}}{2}-S(x) \le \tfrac{V_{i,j}-S(x)}{2}$ due to
Lemma~\ref{lem:bucketbound}; if $K_{i,j} \ne x$, then $\hat{S}_i(x)-S(x) =
\tfrac{V_{i,j}-C_{i,j}}{2}-S(x) \le \tfrac{V_{i,j}}{2}-S(x) \le
\tfrac{V_{i,j}-S(x)}{2}$. 

Combining both cases, we have 
$\Pr[\hat{S}_i(x)-S(x)\ge\tfrac{\epsilon \CS}{2}] \le
\Pr[\tfrac{V_{i,j}-S(x)}{2}\ge\tfrac{\epsilon \CS}{2}] \le \tfrac{1}{2}$ due
to Equation~(\ref{eq:markov}). 

Since $\hat{S}(x)$ is the minimum of all row estimates, we have
$\Pr[\hat{S}(x) \le  S(x) + \tfrac{\epsilon \CS}{2}] 
= 1- \Pr[\hat{S}(x) - S(x) \ge \tfrac{\epsilon \CS}{2}]
= 1- \Pr[\hat{S}_i(x)-S(x) \ge \tfrac{\epsilon \CS}{2}, \forall i]
\ge 1 - (\tfrac{1}{2})^{r} = 1 - \delta$.
\end{proof}

Theorem~\ref{the:hhbound} summarizes the error bounds for heavy hitter
detection in \sysname.

\begin{theorem}
\label{the:hhbound}
\sysname reports every heavy hitter with a probability at least $1-\delta$
(provided that $\phi\CS \ge \epsilon\CS$), and falsely reports a non-heavy
hitter with sum no more than $(\phi - \tfrac{\epsilon}{2})\CS$ with a
probability at most $\delta$.  
\end{theorem}

\begin{proof}
We first prove that \sysname reports each heavy hitter (say $x$) with a high
probability.  If flow $x$ is the majority flow 
in any one of its hashed buckets, it will be reported due to
Lemma~\ref{lem:majority}.  \sysname fails to report $x$ only if $x$ is not the
majority flow of any of its $r$ hashed buckets, i.e., $S(x) \le
\tfrac{V_{i,j}}{2}$ for $1\le i\le r$ and $j=h_i(x)$.  The probability that it
occurs (denoted by $P$) is $P = \Pr[S(x) \le \tfrac{V_{i,j}}{2}, \forall i] =
\Pr[V_{i,j} - S(x) \ge S(x), \forall i]$. Since $S(x) \ge \phi\CS \ge
\epsilon\CS$, we have $P \le \Pr[V_{i,j} - S(x) \ge \epsilon\CS, \forall i]
\le (\tfrac{1}{2})^r = \delta$ due to Equation~(\ref{eq:markov}).  Thus, a
heavy hitter is reported with a probability at least $1-\delta$. 

We next prove that \sysname reports a non-heavy hitter (say $y$) with $S(y)\le 
(\phi - \tfrac{\epsilon}{2})\CS$ with a small probability. 
A necessary condition is that $y$ has its estimate $\hat{S}(y)
\ge \phi\CS$.  Thus, $\hat{S}(y)-S(y)\ge\phi\CS-(\phi -
\tfrac{\epsilon}{2})S=\tfrac{\epsilon\CS}{2}$.  From
Lemma~\ref{lem:pointquery}, we have
$\Pr[\hat{S}(y)-S(y)\ge\tfrac{\epsilon\CS}{2}]\le\delta$.  In other words, $y$
is reported as a heavy hitter with a probability at most $\delta$.
\end{proof}

\begin{table*}[!t]
\centering
\caption{Comparison of \sysname with state-of-the-art invertible sketches.}
\label{tab:compare}
\renewcommand{\arraystretch}{1.1}
\vspace{-1em}
\scriptsize
\begin{tabular}{|c||c|c|c|c|c|c|}
\hline
  Sketch & $r$ & $w$ & FN Prob. & Space & Update time & Detection time \\
\hline
\hline
    Count-Min-Heap & $\log{\tfrac{1}{\delta}}$ &
  $\tfrac{2}{\epsilon}$ & 0 & 
	$O(\tfrac{1}{\epsilon}\log{\tfrac{1}{\delta}}+H\log{n})$ &
    $O(\log{\tfrac{H}{\delta}})$ & $O(H)$ \\
\hline
    LD-Sketch & $\log{\tfrac{1}{\delta}}$ & $\tfrac{2H}{\epsilon}$ & 0 &
	$O(\frac{H}{\epsilon}\log{\tfrac{1}{\delta}}\log{n})$ &
	$O(\log{\tfrac{1}{\delta}})$&
	$O(\frac{H}{\epsilon}\log{\frac{1}{\delta}})$ \\
\hline
    Deltoid & $\log{\tfrac{1}{\delta}}$ & $\tfrac{2}{\epsilon}$ & $\delta$ &
	$O(\frac{1}{\epsilon}\log{\tfrac{1}{\delta}}\log{n})$ &
	$O(\log\tfrac{1}{\delta}\log{n})$ &
	$O(\frac{1}{\epsilon}\log^2{\tfrac{1}{\delta}}\log{n})$ \\
\hline
    Fast Sketch & $4H\log{\tfrac{4}{\delta}}$ &
    $1\!+\!\log\tfrac{n}{4H\log{(4/\delta)}}$ & $\delta$ &
    $O(H\log{\tfrac{1}{\delta}}\log\tfrac{n}{H\log(1/\delta)})$ &
    $O(\log{\tfrac{1}{\delta}}\log\tfrac{n}{H\log{(1/\delta)}})$ &
    $O(H\log^3{\tfrac{1}{\delta}}\log(\tfrac{n}{H\log{(1/\delta)}}))$ \\
\hline
\hline
    \sysname & $\log{\tfrac{1}{\delta}}$ & $\tfrac{2}{\epsilon}$ & $\delta$ &
	$O(\frac{1}{\epsilon}\log{\tfrac{1}{\delta}}\log{n})$ &
	$O(\log{\tfrac{1}{\delta}})$ &
	$O(\frac{1}{\epsilon}\log^2{\tfrac{1}{\delta}})$ \\
\hline
\end{tabular}
\end{table*}

\subsection{Error Bounds for Heavy Changer Detection}

Recall that heavy changer detection relies on the upper bound $U(x)$ and
the lower bound $L(x)$ of $S(x)$ (Section~\ref{subsec:detection}).  From
Lemma~\ref{lem:bucketbound}, both $U(x)$ and $L(x)$ are the true upper and
lower bounds of $S(x)$, respectively.  Lemma~\ref{lem:pointquery} has shown
that $U(x)$, which equals $\hat{S}(x)$, differs from $S(x)$ by a small range
with a high probability.  Now, Lemma~\ref{lem:lowerbound} shows that $L(x)$
and $S(x)$ also differ by a small range with a high probability. 

\begin{lemma}
\label{lem:lowerbound}
$S(x) - L(x) \le \epsilon \CS$ with a probability at least $1-\delta$.
\end{lemma}

\begin{proof}
Consider the lower bound estimate $L_i(x)$ given by the hashed bucket $B(i,j)$
of flow $x$ (where $1\le i\le r$) (Section~\ref{subsec:detection}).  If
$K_{i,j}$ equals $x$, $L_i(x)=C_{i,j}$. By Lemma~\ref{lem:bucketbound}, we
have $S(x) \le \tfrac{V_{i,j}+C_{i,j}}{2}$, implying that $S(x) - L_i(x) =
S(x) - C_{i,j} \le V_{i,j} - S(x)$.

If $K_{i,j} \ne x$, $L_i = 0$ and $x$ is not the majority flow for bucket
$B(i,j)$.  We have $S(x) - L_i(x) = S(x) \le V_{i,j}-S(x)$.  

Combining both cases, we have $\Pr[S(x)-L(x)\ge\epsilon\CS] =
\Pr[S(x)-L_i(x)\ge\epsilon\CS, \forall i] \le 
\Pr[V_{i,j}-S(x)\ge\epsilon\CS, \forall i] \le (\tfrac{1}{2})^r = \delta$ due
to Equation~(\ref{eq:markov}).  
\end{proof}

Lemma~\ref{lem:changebound} provides an upper bound of the estimated maximum
change $\hat{D}(x) = \max\{|U^1(x) - L^2(x)|, |U^2(x) - L^1(x)|\}$ in terms of
$\CS^1$ and $\CS^2$, which are the total sums of all flows in the previous and
current epochs, respectively. 

\begin{lemma}
\label{lem:changebound}
$\hat{D}(x) \le D(x) + \epsilon(\CS^1+\CS^2)$ with a probability at least
$(1-\delta)^2$. 
\end{lemma}

\begin{proof}
Without loss of generality, we consider $\hat{D}(x)=|U^1(x)-L^2(x)|$.  Let
$S^1(x)$ and $S^2(x)$ be the sums of $x$ in the previous and current
epochs, respectively. Let $e_u^1(x) = U^1(x) - S^1(x)$ and $e_l^2 = S^2(x) -
L^2(x)$.  Then $\hat{D}(x)= |S^1(x)+e_u^1(x)-(S^{2}(x)-e_l^2(x))| \le
D(x)+e_u^1(x)+e_l^2(x)$.  Since $e_u^1(x)$ and $e_l^2(x)$ are independent,
we have $\Pr[e_u^1(x)+ e_l^2(x) \le \epsilon(\CS^1+\CS^2)] \ge \Pr[e_u^1(x)
\le \epsilon\CS^1]\cdot\Pr[e_l^2(x) \le \epsilon\CS^2] \ge (1-\delta)^2$,
where the last inequality is due to Lemmas~\ref{lem:pointquery} and
\ref{lem:lowerbound}.  Thus, $\Pr[\hat{D}(x)-D(x)\le\epsilon(\CS^1+\CS^2)]
\ge \Pr[e_u^1(x)+e_l^2(x)\le\epsilon(\CS^1+\CS^2)]\ge(1-\delta)^2$. 
\end{proof}

Theorem~\ref{the:hcbound} summarizes the error bounds for heavy changer
detection in \sysname. 
 
\begin{theorem}
\label{the:hcbound}
\sysname reports every heavy changer with a probability at least $1-\delta$
(provided that $\tfrac{\phi\CD}{\epsilon}\ge\max\{\CS^1,\CS^2\}$), and falsely
reports any non-heavy changer with change no more than $\phi \CD -
\epsilon(\CS^1+\CS^2)$ with a probability at most $1-(1-\delta)^2$.  
\end{theorem}

\begin{proof}
We first prove that \sysname reports each heavy changer (say $x$) with a high
probability. If flow $x$ is the majority
flow in any one of its hashed buckets, it must be reported, as its
estimate $\hat{D}(x) \ge D(x) \ge \phi\CD$.  Flow $x$ is not reported only if
it is not stored as a candidate heavy flow in both sketches.  Since there must
exist one sketch (either in the previous or current epoch) with
$S(x)\ge\phi\CD$, by Theorem~\ref{the:hhbound}, the probability that $x$ is
not reported in that sketch is at most $\delta$ (assuming that 
$\tfrac{\phi\CD}{\epsilon} \ge \max\{\CS^1, \CS^2\}$).  Thus, a heavy changer
is reported with a probability at least $1-\delta$. 

We next prove that \sysname reports a non-heavy changer (say $y$) with 
$D(y)\le\phi \CD-\epsilon(\CS^1+\CS^2)$ with a small probability.
Let $\hat{D}(y) = D(y) + \Delta$ for some $\Delta$; hence,
$\hat{D}(y) \le \phi\CD - \epsilon(\CS^1+\CS^2) + \Delta$.  If $y$ is
reported as a heavy changer, it requires that $\Delta \ge
\epsilon(\CS^1+\CS^2)$ and such a probability is at most $1-(1-\delta)^2$ due
to Lemma~\ref{lem:changebound}.
\end{proof}

\subsection{Error Bounds for Scalable Heavy Flow Detection}

We generalize the analysis for a single detector in Theorems~\ref{the:hhbound}
and \ref{the:hcbound} for scalable heavy flow detection under \sysname.
Our analysis assumes that the stream of packets of each flow is uniformly
distributed to $d \le q$ detectors.  

\begin{theorem}
\label{the:dis_hhbound}
The controller reports every heavy hitter with a probability at least
$(1-\delta)^d$, and falsely reports a non-heavy hitter with sum no more than
$\tfrac{d}{q}(\phi - \tfrac{\epsilon}{2})\CS$ with a probability at most
$1-(1-\delta)^d$. 
\end{theorem}

\begin{proof}
We first study the probability of reporting each heavy hitter (say
$x$).  Recall that the estimate of $x$ at each detector is at least
$\frac{\phi}{d}\CS$ (Section~\ref{subsec:theory_hh}). If all $d$ detectors
report flow $x$ to the controller, flow $x$ must be reported as a heavy hitter
since its added estimate is at least $d\times\frac{\phi}{d}\CS = \phi\CS$.
Such a probability is at least $(1-\delta)^d$ by Theorem~\ref{the:hhbound}. 
	
We next study the probability of reporting a non-heavy hitter (say
$y$). It happens if at least one detector reports flow $y$ to the controller.
If $S(y)\le\tfrac{d}{q}(\phi-\tfrac{\epsilon}{2})\CS$, the sum of flow $y$ at
each detector is at most $\tfrac{1}{q}(\phi-\tfrac{\epsilon}{2})\CS$.  From
Theorem~\ref{the:hhbound}, the probability that a detector reports flow $y$ is
at most $\delta$.  Thus, it is falsely reported as a heavy hitter by the
controller with a probability at most $1- (1-\delta)^d$. 
\end{proof}

\begin{theorem}
\label{the:dis_hcbound}
The controller reports every heavy changer with a probability at least
$(1-\delta)^d$, and falsely reports a non-heavy changer with change no more
than $\tfrac{d}{q}(\phi\CD - \epsilon(\CS^1+\CS^2))$ with a probability at
most $1-(1-\delta)^{2d}$. 
\end{theorem}

\begin{proof}
It is similar to that in Theorem~\ref{the:dis_hhbound} and omitted. 
\end{proof}

\subsection{Error Bounds for Network-Wide Heavy Flow Detection}
\label{subsec:nwbound}

For network-wide heavy flow detection, we can readily check that the
complexity of the merge operation in Algorithm~\ref{alg:merge} is $O(rwq^2) =
O(\tfrac{q^2}{\epsilon}\log{\tfrac{1}{\delta}})$.  In the following, we 
analyze the error bounds for network-wide heavy flow detection. 

Lemma~\ref{lem:nw_majority} shows that if there exists a majority flow
(defined in Section~\ref{subsec:theory_hh}) in a bucket of the merged sketch,
then the bucket can track the majority flow, even though each of the detectors
only sees a portion of traffic of the majority flow. 

\begin{lemma}
\label{lem:nw_majority}
After the $q$ MV-Sketch instances are merged, if there exists a majority flow
$x$ in $B(i,j)$ in the merged sketch, then it must be stored in $K_{i,j}$.     
\end{lemma}

\begin{proof}
We first show that $x$ is stored in $K^k_{i,j}$ in at least one $B^k(i,j)$ for
$1\le k \le q$.  Suppose that the contrary holds. Then by
Lemma~\ref{lem:bucketbound}, 
$S(x)\le \sum_{k}\tfrac{V^k_{i,j} - C^k_{i,j}}{2}
\le \tfrac{V_{i,j}}{2}$,
which contradicts the definition of a majority flow. 

We next show that $x$ must be the key with the maximum network-wide estimate
among all $K^k_{i,j}$'s for $1\le k \le q$.  Suppose the contrary that 
    $y\neq x$ is the maximum key being returned.  Without loss of generality, let
$K_{i,j}^k = x$ for $1\le k\le t$ for some $t\ge 1$.  By
Algorithm~\ref{alg:merge}, the network-wide estimates of $x$ and $y$ are
$e(x)=\textstyle\sum_{1\le k\le t}\!\tfrac{V^k_{i,j} + C^k_{i,j}}{2} +
\textstyle\sum_{t+1\le k\le q}\!\tfrac{V^k_{i,j} - C^k_{i,j}}{2}$ and 
    $e(y)\le$ $\textstyle\sum_{1\le k\le t}\!\tfrac{V^k_{i,j}- C^k_{i,j}}{2} + 
\textstyle\sum_{t+1\le k\le q}\!\tfrac{V^k_{i,j} + C^k_{i,j}}{2}$, respectively.
    Thus, $e(x) + e(y) \le \sum_k V^k_{i,j} = V_{i,j}$.

This implies that $V_{i,j} \ge 2e(x)$ as $y$ is the maximum key.   By
Lemma~\ref{lem:bucketbound}, $e(x)$ is an upper bound of $S(x)$, but $2S(x) >
V_{i,j}$ as $x$ is a majority flow.  This leads to a contradiction. 
\end{proof}

Lemma~\ref{lem:nw_bound} bounds the sum $S(x)$ of flow $x$ in the merged
sketch, where $S(x)$ now corresponds to the network-wide sum. 

\begin{lemma}
\label{lem:nw_bound}
Consider a bucket $B(i,j)$ that flow $x$ is hashed to in the merged sketch. If
$K_{i,j}$ equals $x$, then $C_{i,j}\le S(x)\le \tfrac{V_{i,j}+C_{i,j}}{2}$;
otherwise, $0\le S(x)\le \tfrac{V_{i,j}-C_{i,j}}{2}.$
\end{lemma}

\begin{proof}
Suppose $K_{i,j}\! =\! x$.  Without loss of generality, let $K_{i,j}^k = x$
for $1\!\le\! k\!\le\! t$ for some $t\!\ge\! 1$.  By
Algorithm~\ref{alg:merge}, 
$2e(x)-V_{i,j}\!=\!
2(\textstyle\sum_{1\le k\le t}\!\tfrac{V^k_{i,j} + C^k_{i,j}}{2}$ 
$+ \textstyle\sum_{t+1\le k\le q}\!\tfrac{V^k_{i,j}\! -\! C^k_{i,j}}{2}) -
V_{i,j}
\!=\! \textstyle\sum_{1\le k\le t}C^k_{i,j}$ 
$- \textstyle\sum_{t+1\le k\le q}C^k_{i,j}$
$\le \textstyle\sum_{1\le k\le t}C^k_{i,j} \le S(x)$ (the last inequality is
due to Lemma~\ref{lem:bucketbound}).  Thus, $C_{i,j} = \max\{2e(x)-V_{i,j},
0\}\le S(x)$.  

Also, by Lemma~\ref{lem:bucketbound}, $e(x)\ge S(x)$.
Since $V_{i,j}$ is the total sum of values in $B(i,j)$, $V_{i,j}\ge S(x)$.
Thus, $\tfrac{V_{i,j}+C_{i,j}}{2}\ge\tfrac{V_{i,j}+2e(x)-V_{i,j}}{2}=e(x)\ge
S(x)$. 
	
Suppose now $K_{i,j}\neq x$. Let $K_{i,j} = y$, meaning that $e(y)$ is the
maximum network-wide estimate among all keys hashed to $B(i,j)$.  Without loss
of generality, let $K_{i,j}^k = y$ for $1\!\le\! k\!\le\! t$ for some
$t\!\ge\! 1$.  Note that $C_{i,j}=\max\{2e(y)-V_{i,j}, 0\}$.  If $C_{i,j}=0$,
then $2e(y)-V_{i,j}\le 0$. Thus, $S(x)\le e(x) \le e(y) \le \tfrac{V_{i,j}}{2}
= \tfrac{V_{i,j}-C_{i,j}}{2}$.  

If $C_{i,j} \neq 0$, then $\tfrac{V_{i,j}-C_{i,j}}{2} \!=\! V_{i,j} \!-\! e(y)
\!=\! \textstyle\sum_{k}\!V^{k}_{i,j}
- \textstyle\sum_{1\le k\le t}\!\tfrac{V^k_{i,j}+C^k_{i,j}}{2}$ 
$- \textstyle\sum_{t+1\le k\le q}\!\tfrac{V^k_{i,j}-C^k_{i,j}}{2}$
$= \textstyle\sum_{1\le k\le t}\!\tfrac{V^k_{i,j}-C^k_{i,j}}{2} $
$+ \textstyle\sum_{t+1\le k\le q}\!\tfrac{V^k_{i,j}+C^k_{i,j}}{2} \ge S(x)$ 
(by Lemma~\ref{lem:bucketbound}). 
\end{proof}

Theorem~\ref{the:nw_summary} summarizes the error bounds for network-wide
heavy hitter and heavy changer detection in the merged \sysname.

\begin{theorem}
\label{the:nw_summary}
The merged \sysname reports every heavy hitter with a probability at least
$1-\delta$ (provided that $\phi\CS \ge \epsilon\CS$), and falsely reports a
non-heavy hitter with sum no more than $(\phi - \tfrac{\epsilon}{2})\CS$ with
a probability at most $\delta$; it reports every heavy changer with a
probability at least $1-\delta$ (provided that
$\tfrac{\phi\CD}{\epsilon}\ge\max\{\CS^1,\CS^2\}$), and falsely reports any
non-heavy changer with change no more than $\phi \CD - \epsilon(\CS^1+\CS^2)$
with a probability at most $1-(1-\delta)^2$.  
\end{theorem}

\begin{proof}
By Lemmas~\ref{lem:nw_majority} and \ref{lem:nw_bound}, the bounds of $S(x)$
in the merged \sysname are the same as if all traffic were processed by a
single \sysname. Thus, the network-wide detection of \sysname achieves the
same accuracy as in local detection. 
\end{proof}

\subsection{Comparison with State-of-the-art Invertible Sketches}
\label{subsec:comparison}

We present a comparative analysis on \sysname and state-of-the-art invertible
sketches, including Count-Min-Heap \cite{Cormode2005}, LD-Sketch
\cite{Huang2014}, Deltoid \cite{Cormode2005deltoid}, and Fast Sketch
\cite{Liu2012}.  In the interest of space, we focus on heavy hitter detection
using a single sketch.  Table~\ref{tab:compare} shows the false negative
probability, and the space and time complexities, in terms of $\epsilon$,
$\delta$, $n$, and $H$ (the maximum number of heavy hitters in an epoch). 

We first study the false negative probability (i.e, the maximum probability of
not reporting a heavy hitter); we study other accuracy metrics in
Section~\ref{sec:evaluation}.  Both Count-Min-Heap and LD-Sketch guarantee
zero false negatives as they are configured to keep all heavy hitters in extra
structures, while \sysname can miss a heavy hitter with a probability at most
$\delta$.  Nevertheless, \sysname achieves almost zero false negatives in our
evaluation based on real traces (Section~\ref{sec:evaluation}).  

Regarding the space complexity, all sketches have a $\log n$ term.  However,
it refers to $\log n$ bits (i.e., the key length) in Count-Min-Heap,
LD-Sketch, and \sysname, while it refers to $\log n$ integer counters in
Deltoid and Fast Sketch. 

Regarding the (per-packet) update time complexity, Count-Min-Heap updates the
sketch ($O(\log\tfrac{1}{\delta})$ time) and accesses its heap if the
packet is from a heavy flow ($O(\log H)$ time), and its update time increases
with $H$.  Both Deltoid and Fast Sketch have high time complexities, which
increase with the key length $\log n$.  Both \sysname and LD-Sketch have the
same update time complexities, yet LD-Sketch may need to expand its
associative arrays on-the-fly and this decreases the overall throughput from
our evaluation (Section~\ref{sec:evaluation}).  

We also present the detection time complexity.  However, our evaluation shows
that the detection time of recovering all heavy flows is very small (within
milliseconds) for all sketches shown in Table~\ref{tab:compare}. 

%--------------------------------------------------------------------------
% Implementation in Programmable Switches
%--------------------------------------------------------------------------
\section{Implementation in Programmable Switches}
\label{sec:tofino}

We study how to deploy \sysname in programmable switches to support heavy flow
detection in the data plane.  However, realizing \sysname with high
performance in programmable switches is non-trivial, due to various
restrictions in the switch programming model.  In this section, we introduce
PISA (Protocol-Independent Switch Architecture), and discuss the challenges of
realizing \sysname in PISA switches.  Finally, we show how we overcome the
challenges to make \sysname deployable. 

\subsection{Basics}
\label{subsec:pisa}

We target a family of programmable switches based on PISA
\cite{Bosshart2013,Sivaraman2016}. A PISA switch consists of a programmable
parser, followed by an ingress/egress pipeline of stages, and finally a
de-parser. Packets are first parsed by the parser, which extracts header
fields and custom metadata to form a {\em packet header vector (PHV)}. The PHV
is then passed to the ingress/egress pipeline of stages that comprises 
{\em match-action tables}. Each stage matches some fields of the PHV with a
list of entries and applies a matched action (e.g., modifying PHV fields,
updating persistent states, or performing routing) to the packet.  Finally,
the de-parser reassembles the modified PHV with the original packet and sends
the packet to an output port. PISA switches are fast in packet forwarding, by
limiting the complexity of stages in the pipeline.  Each stage has its own
dedicated resources, including SRAM and multiple arithmetic logic units (ALUs)
that run in parallel. 

PISA switches achieve programmability by supporting multiple customizable
match-action tables in the same stage and connecting many stages into a
pipeline. Programmers can write a program using a domain-specific language
(e.g., P4 \cite{Bosshart2014}) to define packet formats, build custom
processing pipelines, and configure the match-action tables.  

\subsection{Challenges}
\label{subsec:base}

Supporting heavy flow detection in PISA switches must address the hardware
resource constraints \cite{Sapio2017,Sivaraman2017, Gupta2018}: (i) the SRAM
of each stage is of small and identical size (e.g., few megabytes); (ii) the
number of available ALUs per stage is limited; (iii) the pipeline contains a
fixed number of physical stages (e.g., 1-32); and (iv) only a limited size of
a PHV can be passed across stages (e.g., few kilobits). Nevertheless, the
small and static memory design feature of \sysname makes it a good fit to
address the limited resources in PISA switches. 

However, realizing \sysname in PISA switches still faces programming
challenges, due to the restrictions in the switch programming model.  Consider
the update operation of \sysname in Algorithm~\ref{alg:update}.  For
simplicity, we focus on $r=1$ row in the sketch in the following discussion,
yet we can generalize our analysis for multiple rows by duplicating the
single-row implementation.  Intuitively, we can create three register arrays,
namely $V$, $K$ and $C$, to track the total sum, the candidate flow key, and
the indicator counter in \sysname, respectively.  However, there are several
programming challenges. 

\para{Challenge 1 (C1): Limited computation capability for handling
flow keys with more than 32 bits.}  ALUs of PISA switches now only support
primitive arithmetic (e.g., addition and subtraction) on the variables of up
to 32 bits.  While \sysname is also designed based on primitive arithmetic
only, updating the candidate flow key in $K$ is beyond the capability of the
ALUs since the size of a flow key is typically more than 32 bits (e.g., a
5-tuple flow key has 104 bits). 

\para{Challenge 2 (C2): Limited memory access for managing dependent
fields.}  PISA switches support a limited memory access model.  First, the
time budget for each memory access is limited, as only one read-modify-write
is allowed for each variable.  Also, each memory block can only be accessed in
the stage to which it belongs, meaning that we can only access a memory region
{\em once} as a packet traverses the pipeline.  Note that PISA switches allow
concurrent memory accesses to mutually exclusive memory blocks in a single
stage. However, in Algorithm~\ref{alg:update}, the operations on $K$ and $C$
are dependent on each other: the write to $C$ is conditioned on $K$ (Line~4),
while the write to $K$ is conditioned on $C$ (Line~8).  If we want to update
$K$ and $C$ in one stage, we need to perform multiple reads and writes
sequentially, which breaks the time budget of a stage; however, if we place
$K$ and $C$ in different stages, $K$ needs to be accessed twice (the first
access is to check the content of $K$ in Line~3, and the second access is to
update $K$ in Line~8). 

\para{Challenge 3 (C3): Limited branching for updates.}  To simplify
processing, PISA switches design their ALUs with a small circuit depth (e.g.,
3) that hinders complicated predicted operations.  The packet processing in a stage
typically supports an if-else chain with at most two levels
\cite{Sivaraman2016}. In Algorithm~\ref{alg:update}, updating $C$ requires a
three-level if-else chain (Lines~4, 6, and 9), which makes it difficult to
perform the update operation within one stage.  While complex branching is
allowed across stages, updating $C$ in different stages is infeasible with the
memory access model of PISA switches (C2).  

\subsection{Implementation}
\label{subsec:generalize}

\begin{algorithm}[t]
\caption{Implementation for 5-tuple flow keys}
\label{alg:mvgen}
\begin{small}
\begin{algorithmic}[1]
\If {Metadata.repass = 0}
\State // Stage 1: update $V_{h(x)}$ and access $(K^1_{h(x)}, K^2_{h(x)})$ 
\State $V_{h(x)} \leftarrow V_{h(x)} + v_{x}$    
    \If {$(K^1_{h(x)}, K^2_{h(x)})\ne (x_1, x_2)$}
    \State Metadata.flag $\leftarrow$ 1
\EndIf
    \State  // Stage 2: access $(K^3_{h(x)}, K^4_{h(x)})$ 
    \If {$(K^3_{h(x)}, K^4_{h(x)})\ne (x_3, x_4)$}
\State Metadata.flag $\leftarrow$ 1
\EndIf
    \State // Stage 3: update $C_{h(x)}$  
    \If {Metadata.flag $\ne 1$ }
    \State $C_{h(x)} \leftarrow C_{h(x)} + v_{x}$
    \ElsIf {Metadata.flag $=1$ and  $C_{h(x)} \ge v_{x}$}
    \State $C_{h(x)} \leftarrow C_{h(x)} - v_{x}$
    \EndIf 
    \If {Metadata.flag $=1$ and $C_{h(x)} < v_{x}$}
        \State Metadata.repass $\leftarrow$ 1
    \EndIf
    \State // Stage 4: recirculate $x$ 
    \If {Metadata.repass $=1$}
        \State recirculate $x$ 
    \EndIf
\Else 
    \State // Stage 1: update $(K^1_{h(x)}, K^2_{h(x)})$
    \State $(K^1_{h(x)}, K^2_{h(x)}) \leftarrow (x_1, x_2)$
    \State // Stage 2: update $(K^3_{h(x)}, K^4_{h(x)})$ 
    \State $(K^3_{h(x)}, K^4_{h(x)}) \leftarrow (x_3, x_4)$
    \State // Stage 3: update $C$ 
    \State $C_{h(x)} \leftarrow v_x - C_{h(x)}$
\EndIf
\end{algorithmic}
\end{small}
\end{algorithm}

We elaborate how we address the challenges of implementing \sysname in PISA
switches.  To address Challenge~C1, we split a long flow key into multiple
sub-keys and use multiple stages and ALUs to process the sub-keys. For
example, we can split a 104-bit 5-tuple flow key into three 32-bit sub-keys
and one 8-bit sub-key, and access each sub-key in one stage with a single ALU. 
To reduce the number of stages, we can use {\em paired atoms}
\cite{Sivaraman2016} (an atom refers to a packet-processing unit) to update a
pair of sub-keys in one stage.  Specifically, in paired updates, the ALUs of
PISA switches can read two 32-bit elements from the register memory, set up
conditional branching based on both elements, perform primitive arithmetic,
and write back the final results. 

To address Challenges~C2 and C3, we leverage the {\em recirculation} feature
\cite{Sivaraman2017,Ben2018precision} of PISA switches to eliminate the
inter-dependency between $K$ and $C$ and the complex branching for updating
$C$.  We define the {\em change point} as the point where we need to update
the candidate heavy flow key and negate the indicator counter during the
update process (i.e., Lines~8-9 in Algorithm~\ref{alg:update}).  Our idea is
to put the operations at the change point in the second pass of a packet, such
that the operations are carried out if and only if the packet is recirculated
to the second pass of the switch pipeline.  More concretely, in the first
pass, we just read $K$, update $C$, and recirculate the packet if the change
point appears; in the second pass, we update $K$ and negate $C$. 

Algorithm~\ref{alg:mvgen} shows the pseudo-code of implementing \sysname in
PISA switches for 5-tuple flow keys.  Let $V_{h(x)}$, $K_{h(x)}$, and
$C_{h(x)}$ be the entries that $x$ is hashed into in the register arrays $V$,
$K$, and $C$ via the hash function $h$, respectively.  We split a 104-bit flow
key $x$ into four sub-keys: source IP $x_1$, destination IP $x_2$,
source-destination ports $x_3$, and protocol $x_4$.  We use two register
arrays, $(K^1, K^2)$ and $(K^3, K^4)$, to track the candidate heavy flow
sub-keys, such that each element in the arrays is a pair of 32-bit variables.
We use the metadata {\em Metadata.repass}, initialized as zero, in the first
pass of each packet to control the execution of the pipeline. In the first
pass of $x$, we perform the following operations. In Stage~1 and Stage~2, we
update $V_{h(x)}$ and compare each sub-key with the candidate heavy flow
sub-keys; if either sub-key is not matched, we set {\em Metadata.flag} as one
(Lines~4-10).  We update $C_{h(x)}$ based on the value of {\em Metadata.flag}
in Stage~3 (Lines~11-16). Finally, in Stage~4, we check the value of 
{\em Metadata.repass}; if it is one, we recirculate the packet together with
{\em Metadata.repass} to the switch pipeline (Lines~21-23). In the second
pass of $x$, we update the two register arrays $(K^1, K^2)$ and $(K^3, K^4)$,
as well as $C_{h(x)}$ (Lines~25-30). 

Note that our evaluation (Section~\ref{subsec:evaluation_hw}) shows that
\sysname requires a second pass only on a small fraction of packets (e.g.,
less than 5\%), meaning that the recirculation overhead is limited. 

\subsection{Optimizations}
\label{subsec:smallkey}

\begin{algorithm}[t]
\caption{Size counting on 32-bit flow keys}
\label{alg:mvsc}
\begin{small}
\begin{algorithmic}[1]
\If {Metadata.repass = 0}
\State // Stage 1: update $V_{h(x)}$ and $(K_{h(x)}, C_{h(x)})$
\State $V_{h(x)} \leftarrow V_{h(x)} + v_{x}$    
\If {$K_{h(x)} = x$}
    \State $C_{h(x)} \leftarrow C_{h(x)} + v_{x}$
\ElsIf {$K_{h(x)} \neq x$ and  $C_{h(x)} \ge v_{x}$}
    \State $C_{h(x)} \leftarrow C_{h(x)} - v_{x}$
\EndIf
    \If {$K_{h(x)} \neq x$ and $C_{h(x)} < v_{x}$}
    \State $K_{h(x)} \leftarrow x$
    \State Metadata.repass = 1
\EndIf
\State // Stage 2: recirculate $x$
\If {Metadata.repass = 1} 
    \State recirculate $x$ 
\EndIf
\Else 
\State // Stage 1: this is the second pass
\State $C_{h(x)} \leftarrow v_x - C_{h(x)}$
\EndIf
\end{algorithmic}
\end{small}
\end{algorithm}

\begin{algorithm}[t]
\caption{Packet counting on 32-bit flow keys}
\label{alg:mvpc}
\begin{small}
\begin{algorithmic}[1]
\State // Stage 1: update $V_{h(x)}$ and $(K_{h(x)}, C_{h(x)})$ 
\State $V_{h(x)} \leftarrow V_{h(x)} + v_{x}$    
\If {$K_{h(x)} = x$ or $C_{h(x)} = 0$}
    \State $C_{h(x)} \leftarrow C_{h(x)} + 1$
\Else 
    \State $C_{h(x)} \leftarrow C_{h(x)} - 1$
\EndIf
\If {$K_{h(x)} \neq x$ and $C_{h(x)} = 0$}
    \State $K_{h(x)} \leftarrow x$
\EndIf
\end{algorithmic}
\end{small}
\end{algorithm}

We can optimize the implementation of \sysname if the flow keys have no more
than 32 bits (e.g, the source or destination IPv4 address).  This allows us to
access a flow key via a single ALU (i.e., C1 addressed), and update $K$ and
$C$ atomically via paired atoms to address their dependency (i.e., C2
addressed).  Specifically, we can place $K$ and $C$ in a 64-bit register
array, in which the high 32 bits of each entry store the key field, while the
low 32 bits store the indicator counter field.  The paired atom packs the
operations of reading, conditional branching, primitive arithmetic, and
writing for both $K$ and $C$ atomically. 

We show how to update the pair ($K$, $C$) via limited branching in two cases:
(i) {\em size counting}, which counts the total bytes of each flow; and (ii)
{\em packet counting}, which counts the number of packets of each flow. For
size counting, similar to Algorithm~\ref{alg:mvgen}, we put the update of
$C$ at the change point in the second pass of the packet.  For packet
counting, the size $v_{x}$ of flow $x$ is the constant one.  We observe that
when the packet processing reaches the change point, the state of $C_{h(x)}$
is zero and should change from zero to one (Lines~6-9 in
Algorithm~\ref{alg:update}). It is equivalent to incrementing $C_{h(x)}$ by
one as in the case when $K_{h(x)}$ equals $x$ (Line~4 in
Algorithm~\ref{alg:update}). By merging these two branches, we can change and
reorganize the if-conditions in Algorithm~\ref{alg:update} to reduce the
three-level if-else chain to a two-level one, as well as eliminate the
recirculation operations (i.e., C3 addressed). 

Algorithms~\ref{alg:mvsc} and \ref{alg:mvpc} summarize our optimized
implementation of \sysname for size counting and packet counting in PISA
switches, respectively.  Note that Lines~4-12 of Algorithm~\ref{alg:mvsc} and
Lines~3-10 of Algorithm~\ref{alg:mvpc} can be done in one paired atom. 

%---------------------------------------------------------------------------
% Evaluation
%---------------------------------------------------------------------------
\section{Evaluation}
\label{sec:evaluation} 

We conduct evaluation in both software and hardware environments. Our
trace-driven evaluation in software shows that \sysname achieves (i) high
accuracy in heavy flow detection with small and static memory space, (ii) high
processing speed, and (iii) high accuracy in scalable detection, compared to
state-of-the-art invertible sketches.  We also show how SIMD instructions can
further boost the update performance of \sysname in software.  Furthermore, 
our evaluation in a Barefoot Tofino switch \cite{tofino} shows that \sysname
achieves (i) line speed for packet counting and incurs slight (e.g., less than
5\%) performance degradation for size counting, and (ii) incurs only limited
switch resource overhead.  

\subsection{Evaluation in Software}
\label{subsec:evaluation_sw}

\para{Simulation testbed.} We conduct our evaluation on a server equipped with
an eight-core Intel Xeon E5-1630 3.70\,GHz CPU and 16\,GB RAM. The CPU has
64\,KB of L1 cache per core, 256\,KB of L2 cache per core, and 10\,MB of
shared L3 cache.  The server runs Ubuntu 14.04.5.  To exclude the I/O overhead
on performance, we load all datasets into memory prior to all experiments. 

\para{Dataset.} 
We use the anonymized Internet traces from CAIDA \cite{caida}, captured on an
OC-192 backbone link in April 2016.  The original traces are one hour long,
and we focus on the first five minutes of the traces in our evaluation.  We
divide the traces into five one-minute epochs and obtain the average results.
We measure IPv4 packets only.  Each epoch contains 29\,M packets, 1\,M flows,
and 6\,M unique IPv4 addresses on average. 

\para{Methodology.} We take the source/destination address pairs as flow
keys (64 bits long).  For evaluation purposes, we generate the ground truths
by finding $\CS$ and $\CD$, and hence the true heavy flows, for different
epochs.   We use MurmurHash \cite{murmurhash} as the hash function in all
sketches.

We compare \sysname (MV) with state-of-the-art invertible sketches, including
Count-Min-Heap (CMH) \cite{Cormode2005}, LD-Sketch (LD) \cite{Huang2014},
Deltoid (DEL) \cite{Cormode2005deltoid}, and Fast Sketch (FAST)
\cite{Liu2012}. We do not consider Reversible Sketch \cite{Schweller2007} and
SeqHash \cite{Bu2010} (Section~\ref{subsec:sketches}) due to their high update
costs (which increase with the key length) and high enumeration costs in
recovering heavy flows.

We consider various memory sizes for each sketch in our evaluation.  We fix
$r=4$ and vary $w$ according to the specified memory size.  By default, we
choose the threshold that keeps the number of heavy flows detected in each
epoch as 80 on average.  For CMH, we allocate an extra 4\,KB of memory for its
heap data structure to store heavy flows.  For LD, since it dynamically
expands the associative arrays of its buckets (Section~\ref{subsec:sketches}),
we adjust its expansion parameter so that it has comparable memory size to
other sketches. 

\para{Metrics.} We consider the following metrics.
\begin{itemize}[leftmargin=*]
\item \textit{Precision}: fraction of true heavy flows reported over all
reported flows; 
\item \textit{Recall}: fraction of true heavy flows reported over all true
heavy flows; 
\item \textit{F1-score}: $\tfrac{2\times\text{precision}\times\text{recall}}{\text{precision}+\text{recall}}$; 
\item \textit{Relative error}: $\tfrac{1}{|R|}\textstyle\sum_{x, x\in R}\frac{|S(x) - \hat{S}(x)|}{S(x)}$, where $R$ is the set of true heavy
flows reported; and 
\item \textit{Update throughput}: number of packets processed per second
(in units of pkts/s). 
\end{itemize}

\begin{figure}[!t]
\centering
\begin{tabular}{c@{\ }c}
\multicolumn{2}{c}{\includegraphics[width=1.8in]{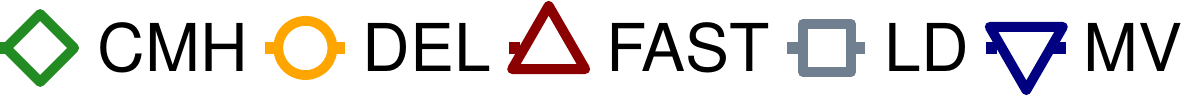}} \\
\includegraphics[width=1.65in]{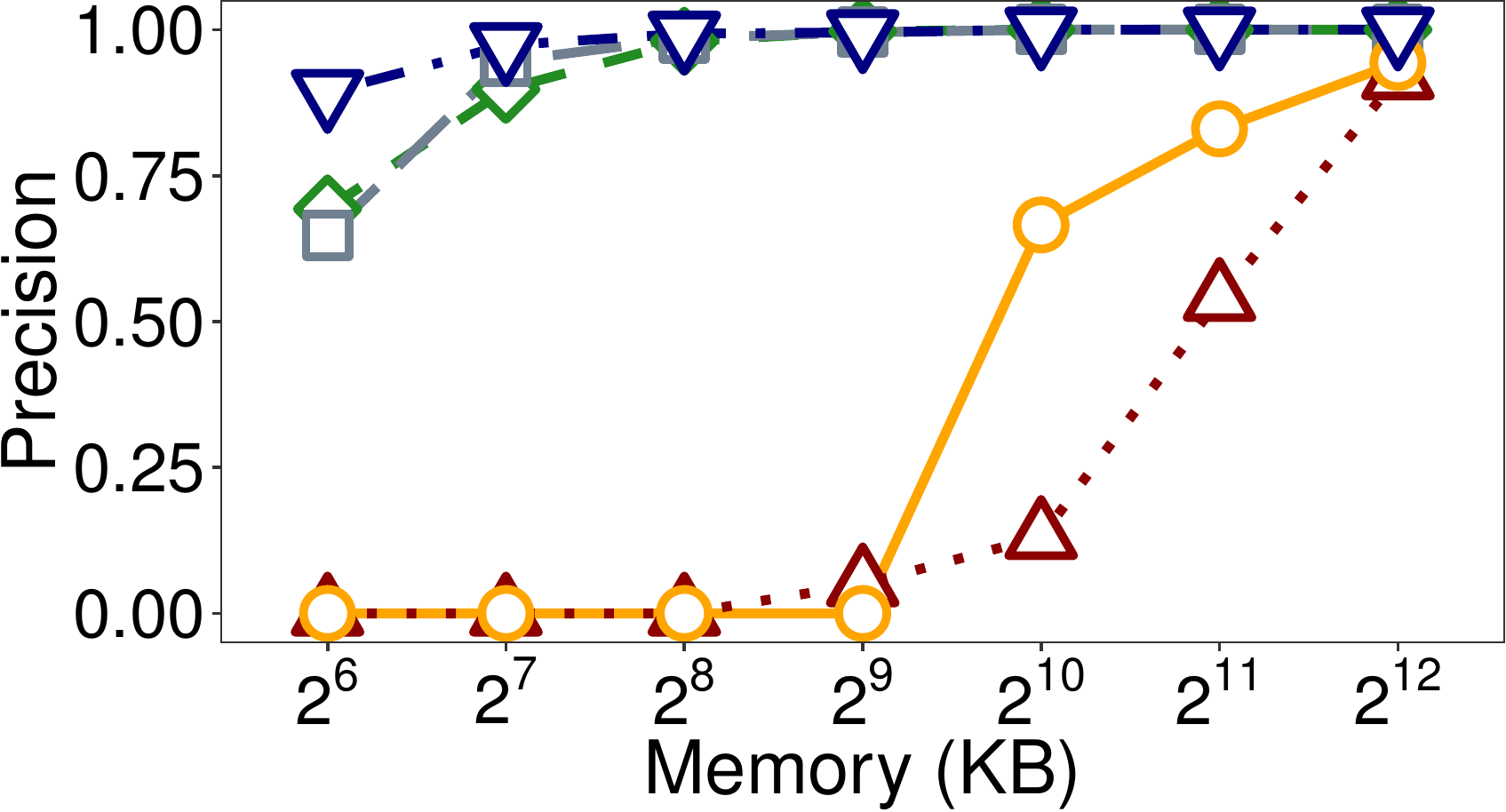} &
\includegraphics[width=1.65in]{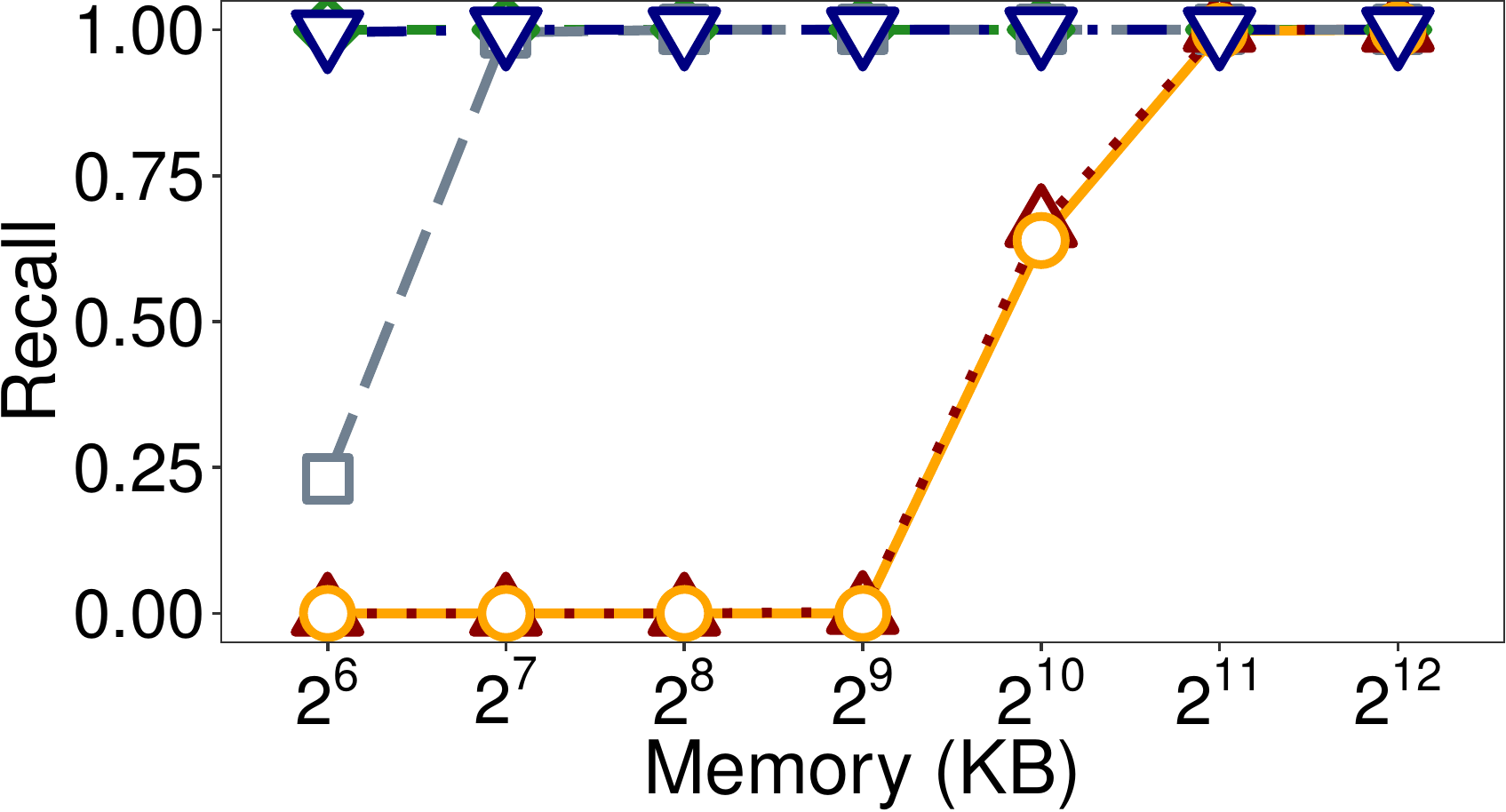} 
\vspace{-3pt}\\
{\footnotesize (a) Precision} & 
{\footnotesize (b) Recall}
\vspace{3pt}\\
\includegraphics[width=1.65in]{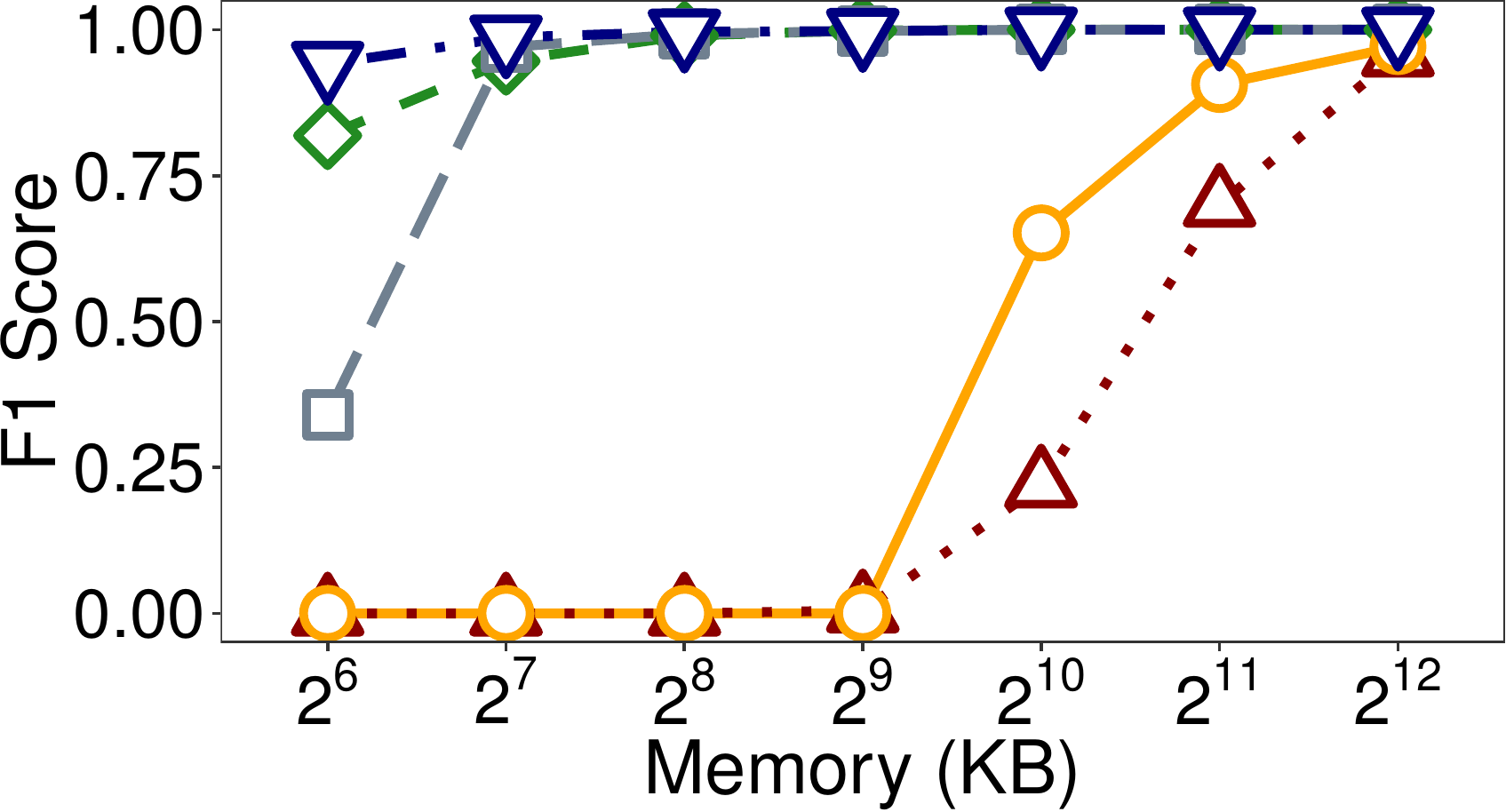} &
\includegraphics[width=1.65in]{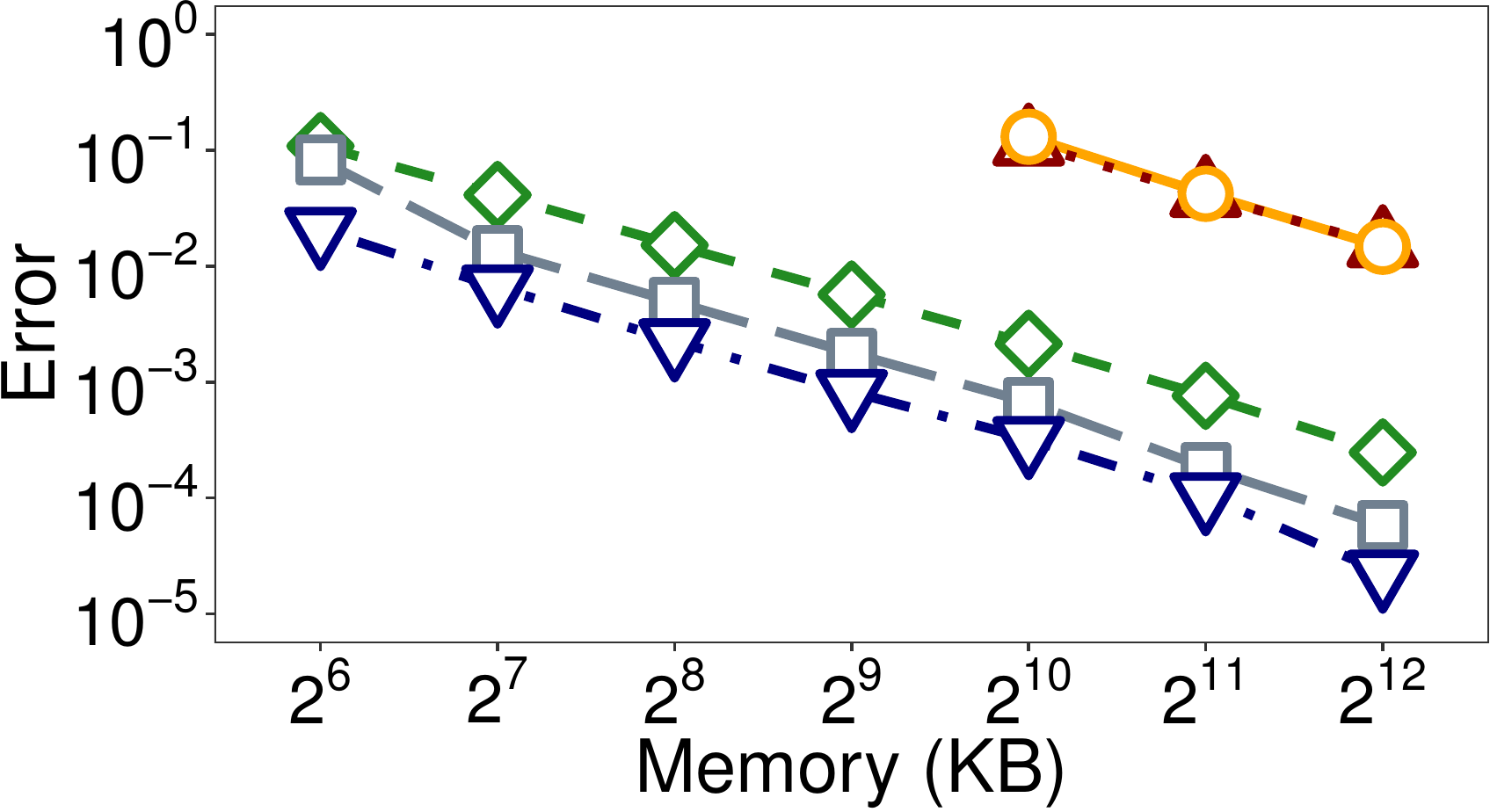} 
\vspace{-3pt}\\
{\footnotesize (c) F1 score} & 
{\footnotesize (d) Relative error} 
\end{tabular}
\vspace{-6pt}
\caption{Experiment 1 (Accuracy for heavy hitter detection). }
\label{fig:exp_hh}
\vspace{-8pt}
\end{figure}

\para{Experiment 1 (Accuracy for heavy hitter detection).}  
Figure~\ref{fig:exp_hh} compares the accuracy of \sysname with that of other
sketches in heavy hitter detection.  Both DEL and FAST have precision and recall
near zero when the amount of memory is 512\,KB or less, as they need more memory
to recover all heavy hitters.  Both CMH and LD have high accuracy except when
the memory size is only 64\,KB, as they incur many false positives
in limited memory. Overall, MV-Sketch achieves high accuracy; for example,
its relative error is on average 55.8\% and 87.2\% less than those of LD and
CMH, respectively. 

\begin{figure}[!t]
\centering
\begin{tabular}{c@{\ }c}
\multicolumn{2}{c}{\includegraphics[width=1.8in]{fig/legend.pdf}} \\
\includegraphics[width=1.65in]{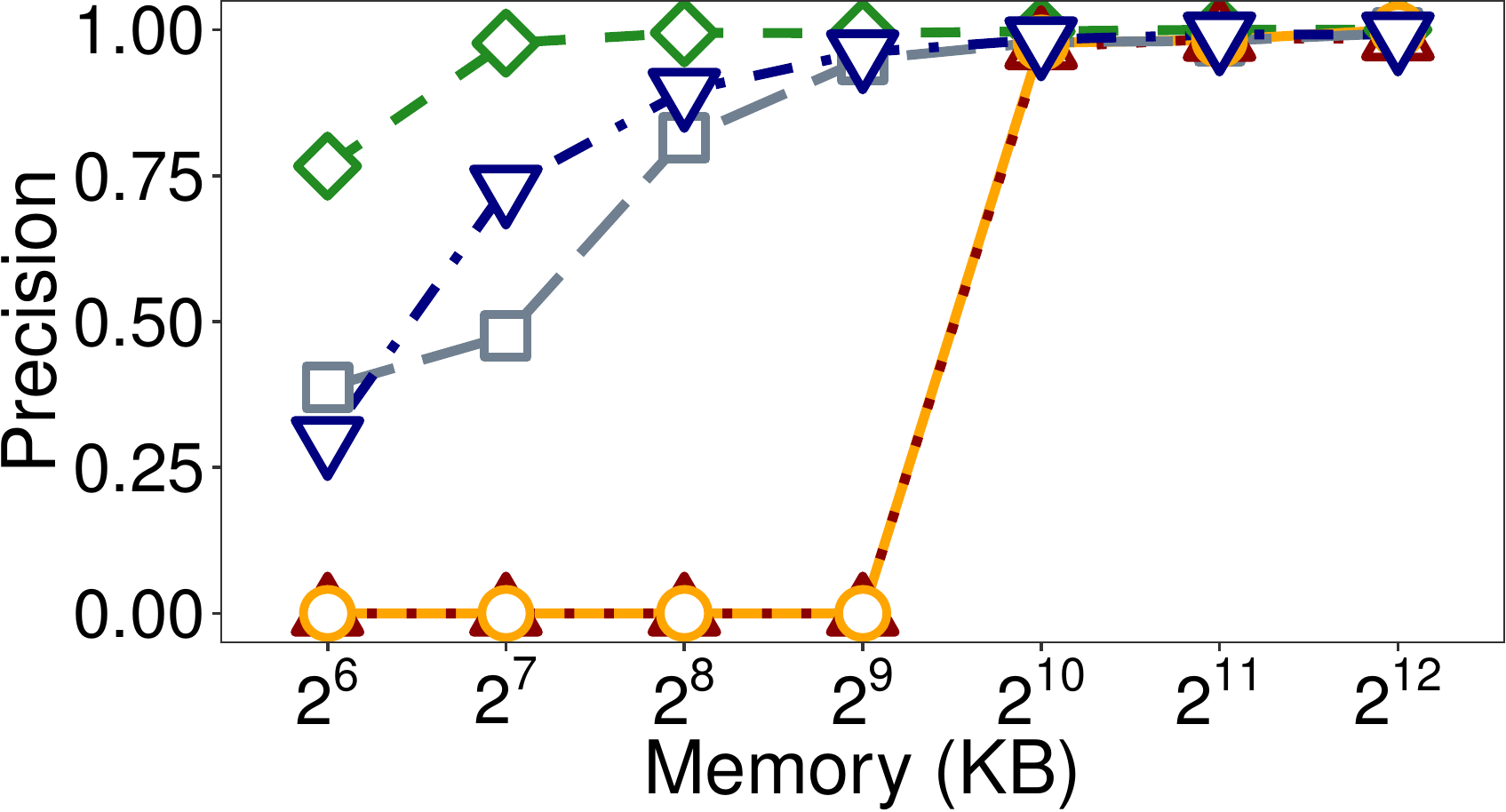} &
\includegraphics[width=1.65in]{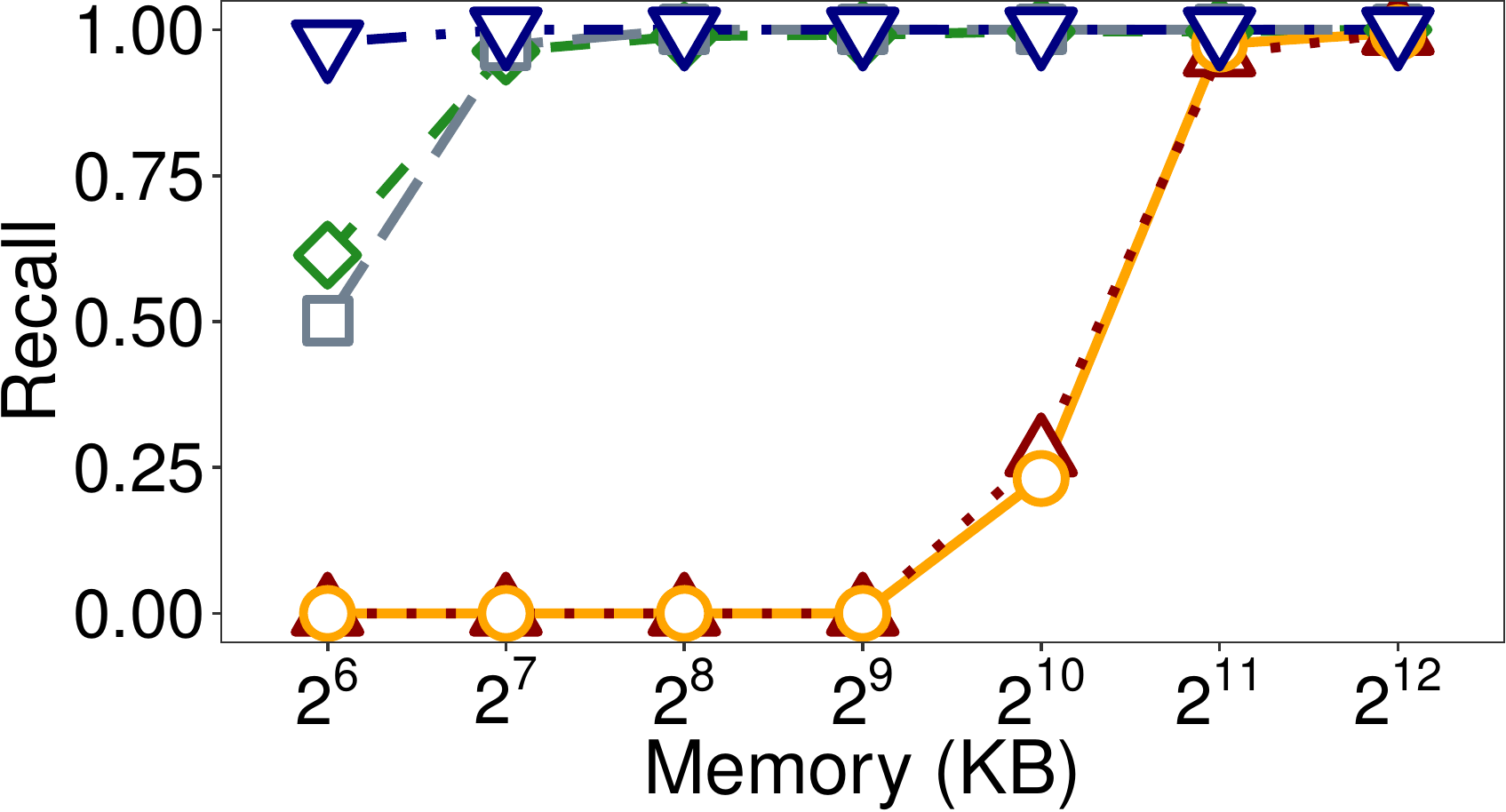} 
\vspace{-3pt}\\
{\footnotesize (a) Precision} & 
{\footnotesize (b) Recall} 
\vspace{3pt}\\
\includegraphics[width=1.65in]{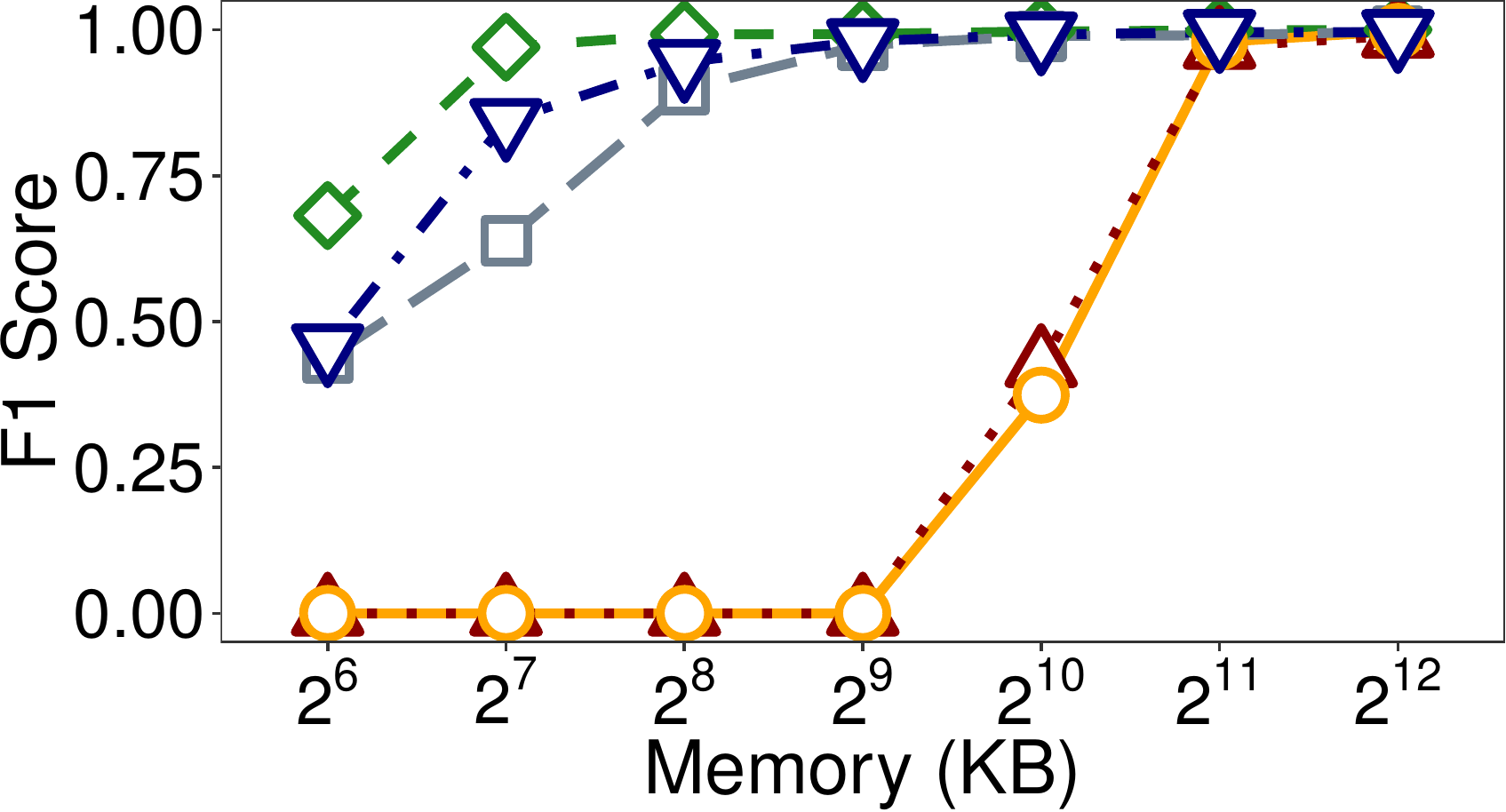} & 
\includegraphics[width=1.65in]{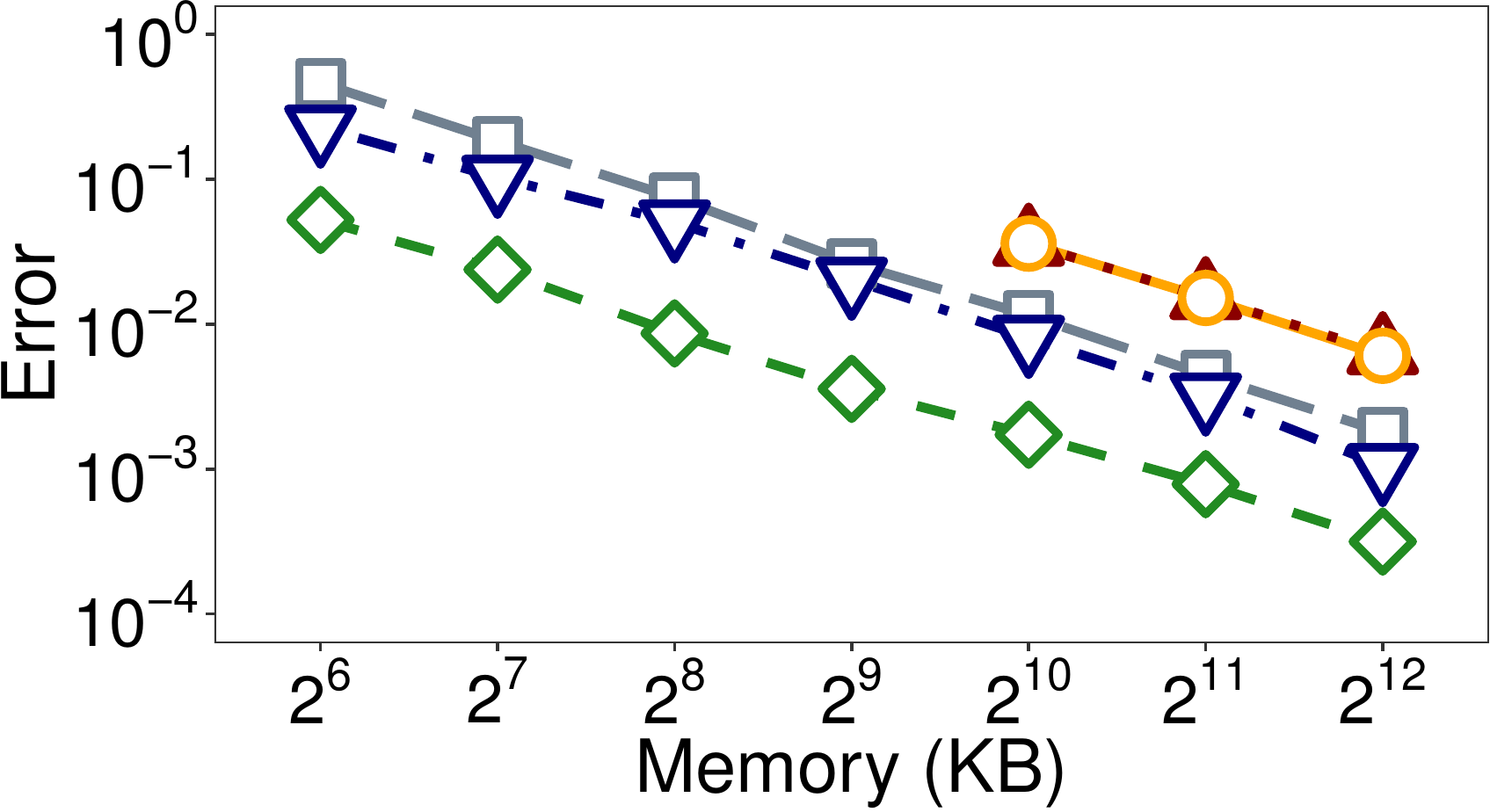} 
\vspace{-3pt}\\
{\footnotesize (c) F1 score} & 
{\footnotesize (d) Relative error} 
\end{tabular}
\vspace{-6pt}
\caption{Experiment 2 (Accuracy for heavy changer detection). }
\label{fig:exp_hc}
\vspace{-8pt}
\end{figure}

\para{Experiment 2 (Accuracy for heavy changer detection).} 
Figure~\ref{fig:exp_hc} compares the accuracy of \sysname with that of 
other sketches in heavy changer detection.  Both DEL and FAST again have
almost zero precision and recall when the memory size is 512\,KB or less.  We
see that CMH has the highest F1 score and smallest relative error among
all sketches, yet its recall is below one for almost all memory sizes. On
the other hand, \sysname maintains a recall of one except when the memory size
is 64\,KB, but its precision is low when the memory size is 256\,KB or less.
The reason is that \sysname uses the estimated maximum change of a flow for
heavy changer detection, thereby having fewer false negatives but more false
positives; we view this as a design trade-off.  \sysname achieves both higher
precision and recall than LD when the memory size is 128\,KB or more. 

\begin{figure}[!t]
\centering
\begin{tabular}{c@{\ }c}
\multicolumn{2}{c}{\includegraphics[width=1.8in]{fig/legend.pdf}} \\
\hspace{-0.1in}
\includegraphics[width=1.65in]{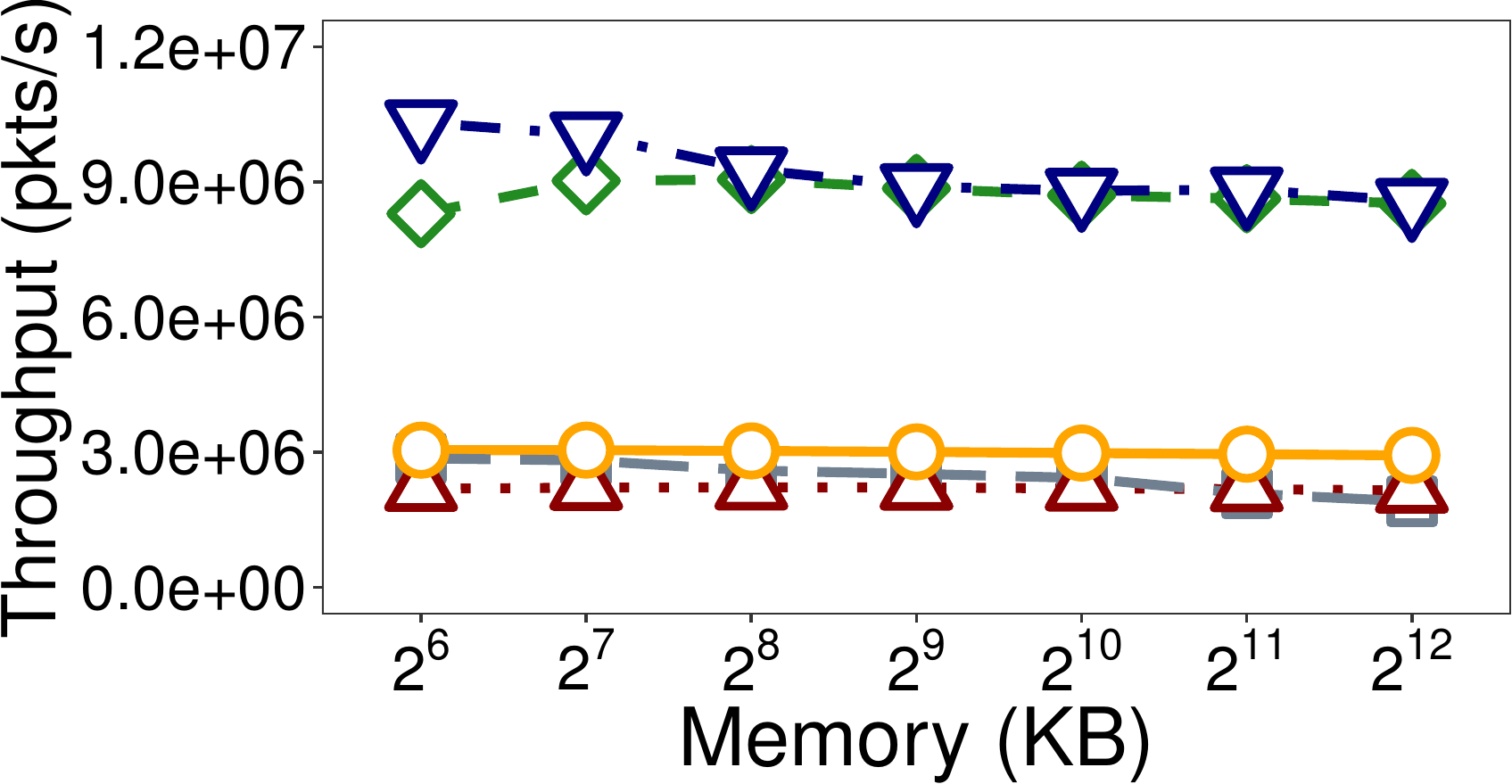} & 
\includegraphics[width=1.65in]{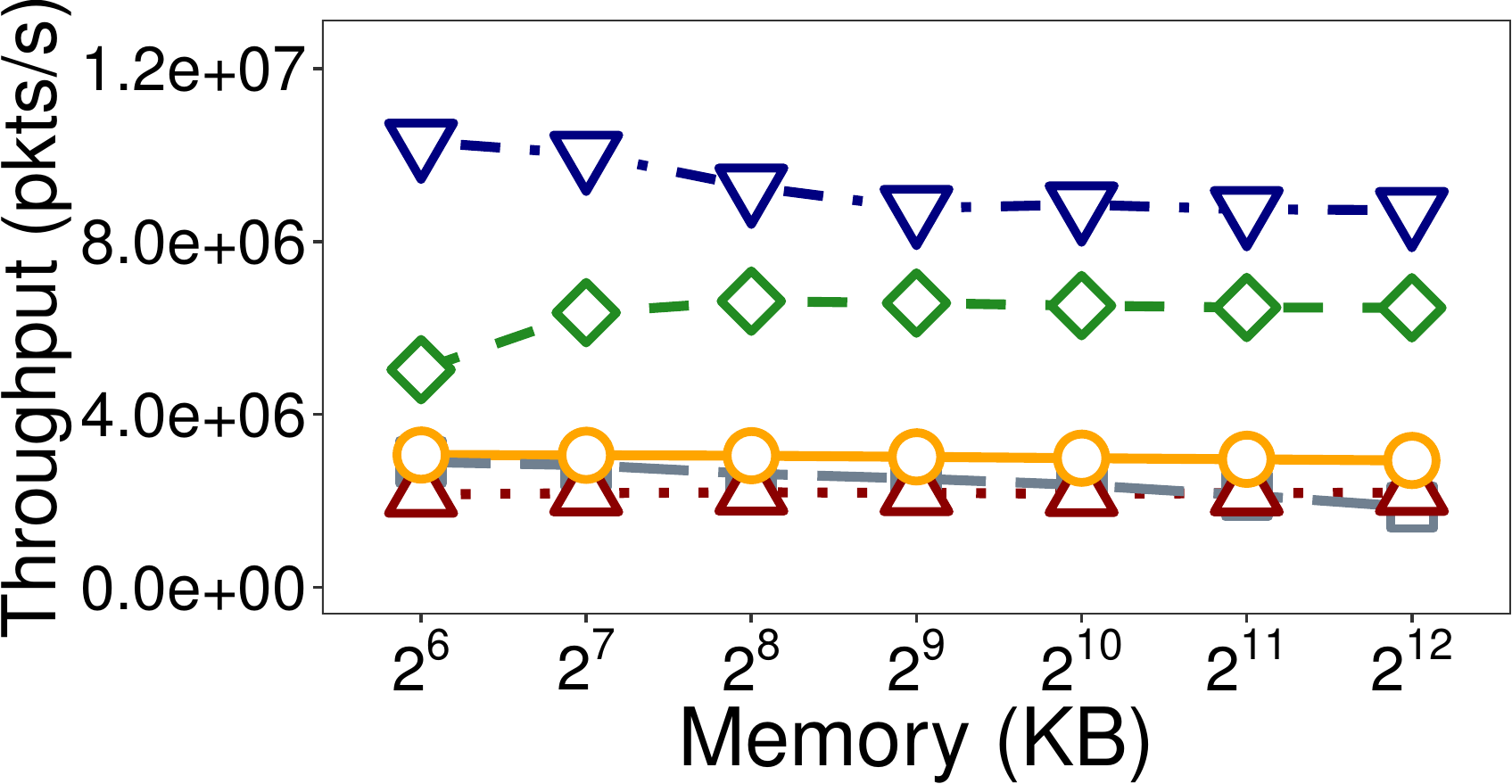} 
\vspace{-3pt}\\
{\footnotesize (a) Heavy hitter detection} &
{\footnotesize (b) Heavy changer detection}
\vspace{3pt}\\
\hspace{-0.1in}
\includegraphics[width=1.65in]{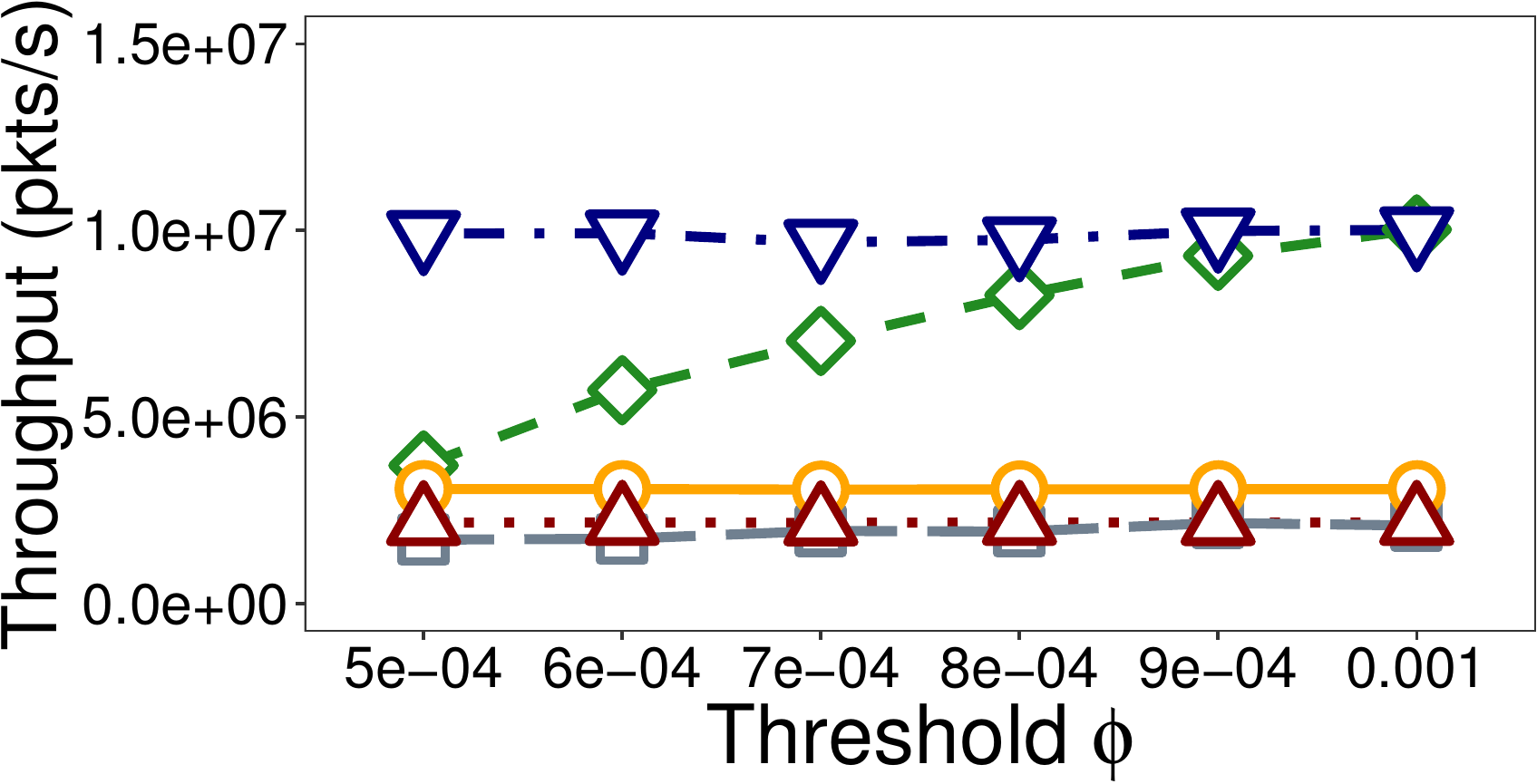} & 
\includegraphics[width=1.65in]{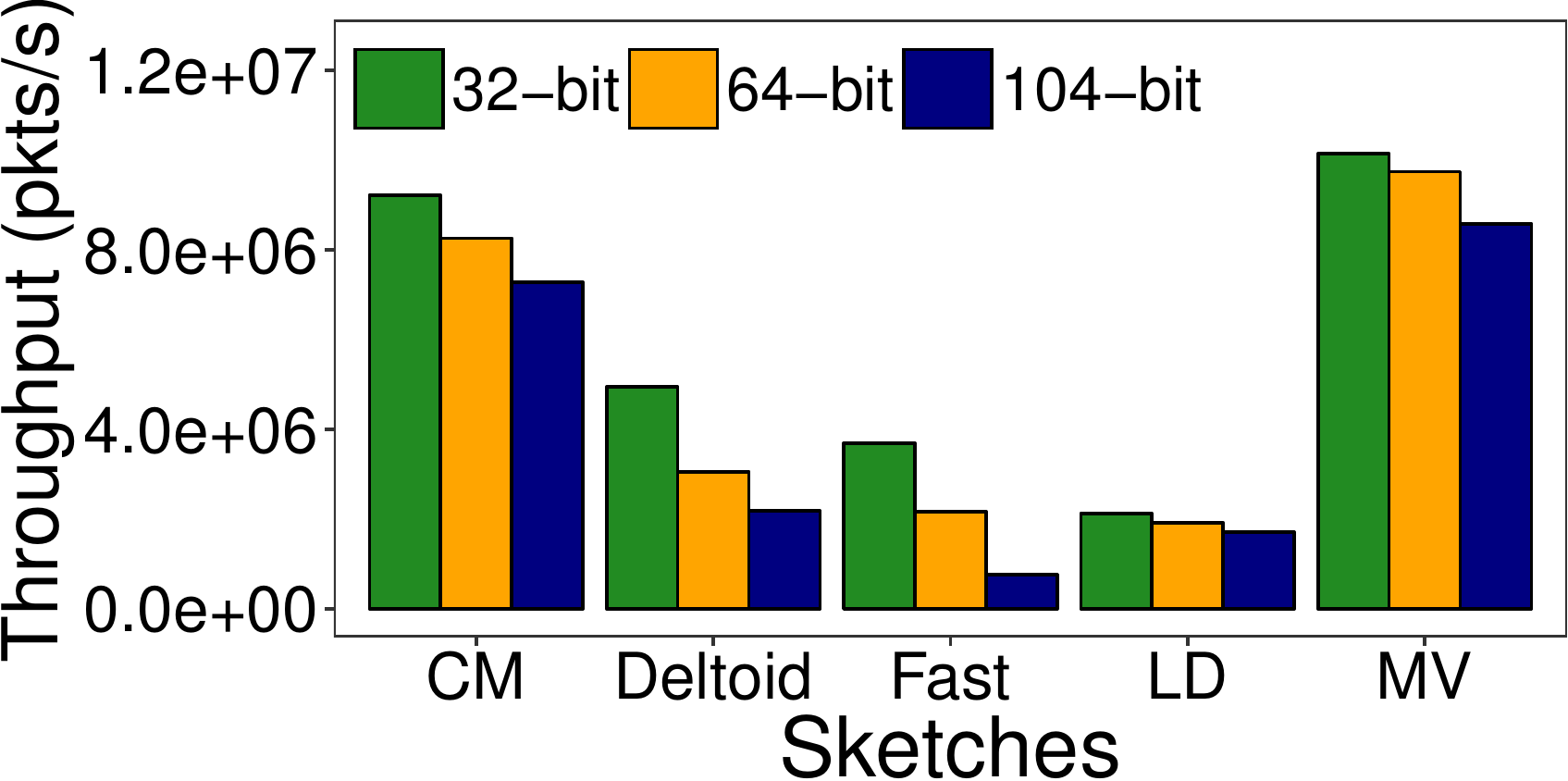}  
\vspace{-3pt}\\
{\footnotesize (c) Impact of threshold $\phi$} &
{\footnotesize (d) Impact of key length} 
\end{tabular}
\vspace{-6pt}
\caption{Experiment 3 (Update throughput).}
\label{fig:exp_hhspeed}
\vspace{-8pt}
\end{figure}

\para{Experiment 3 (Update throughput).}  We now measure the update
throughput of all sketches in different settings.  We present averaged results
over 10 runs.  We omit the error bars in our plots as the variances across
runs are negligible. 

Figure~\ref{fig:exp_hhspeed}(a) shows the update throughput of various
sketches in heavy hitter detection.  \sysname achieves more than $3\times$
throughput over LD, DEL, and FAST, and 24\% higher throughput than CMH when
the memory size is 64\,KB.  Note that \sysname (and other sketches as well)
sees a throughput drop as the memory size increases, since it cannot
be entirely put in cache and the memory access latency increases. The
throughput of CMH is much lower than \sysname, especially when the
memory size is 128\,KB or less, as it sees many false positives and incurs
memory access overhead in its heap. 

Figure~\ref{fig:exp_hhspeed}(b) shows the update throughput of various
sketches in heavy changer detection.  \sysname has the highest throughput,
which is 1.34--2.05$\times$ and 2.98--3.38$\times$ over CMH and other
sketches, respectively.  Note that CMH has lower throughput than in
Figure~\ref{fig:exp_hhspeed}(a) although we keep the same number (i.e., 80) of
heavy flows in both cases.  The reason is that compared to heavy hitter
detection, CMH needs to keep more candidates in the heap to guarantee that all
heavy changers can be found, thereby incurring higher memory access overhead. 

Figure~\ref{fig:exp_hhspeed}(c) shows the impact of the fractional threshold
$\phi$ on the update throughput. Here, we focus on heavy hitter detection and 
fix the memory size as 64\,KB.  \sysname maintains high and stable throughput
(above 9.8\,M pkts/s) regardless of the threshold value.  CMH has
slower throughput for smaller $\phi$ (i.e., more heavy hitters to be
detected).  For example, when $\phi=$~0.0005, the throughput of CMH is
3.7M~pkts/s only.  The reason is that the overhead of maintaining the
heap increases with the number of heavy flows being tracked. 

Figure~\ref{fig:exp_hhspeed}(d) shows the impact of the key length on the
update throughput, by setting the flow keys as source addresses (32 bits),
source/destination address pairs (64 bits), and 5-tuples (104 bits).
We again focus on heavy hitter detection and fix the memory size as 64\,KB.  As
the key length increases from 32 bits to 104 bits, the throughput drops of
\sysname, CMH, and LD are 15-21\%, while those of DEL and FAST are 55-80\%.
The reason is that the numbers of counters in DEL and FAST increase with the
key length, thereby incurring much higher memory access overhead. 

\begin{figure}[!t]
\centering
\begin{tabular}{c@{\ }c}
\multicolumn{2}{c}{\includegraphics[width=1.8in]{fig/legend.pdf}} \\
\includegraphics[width=1.65in]{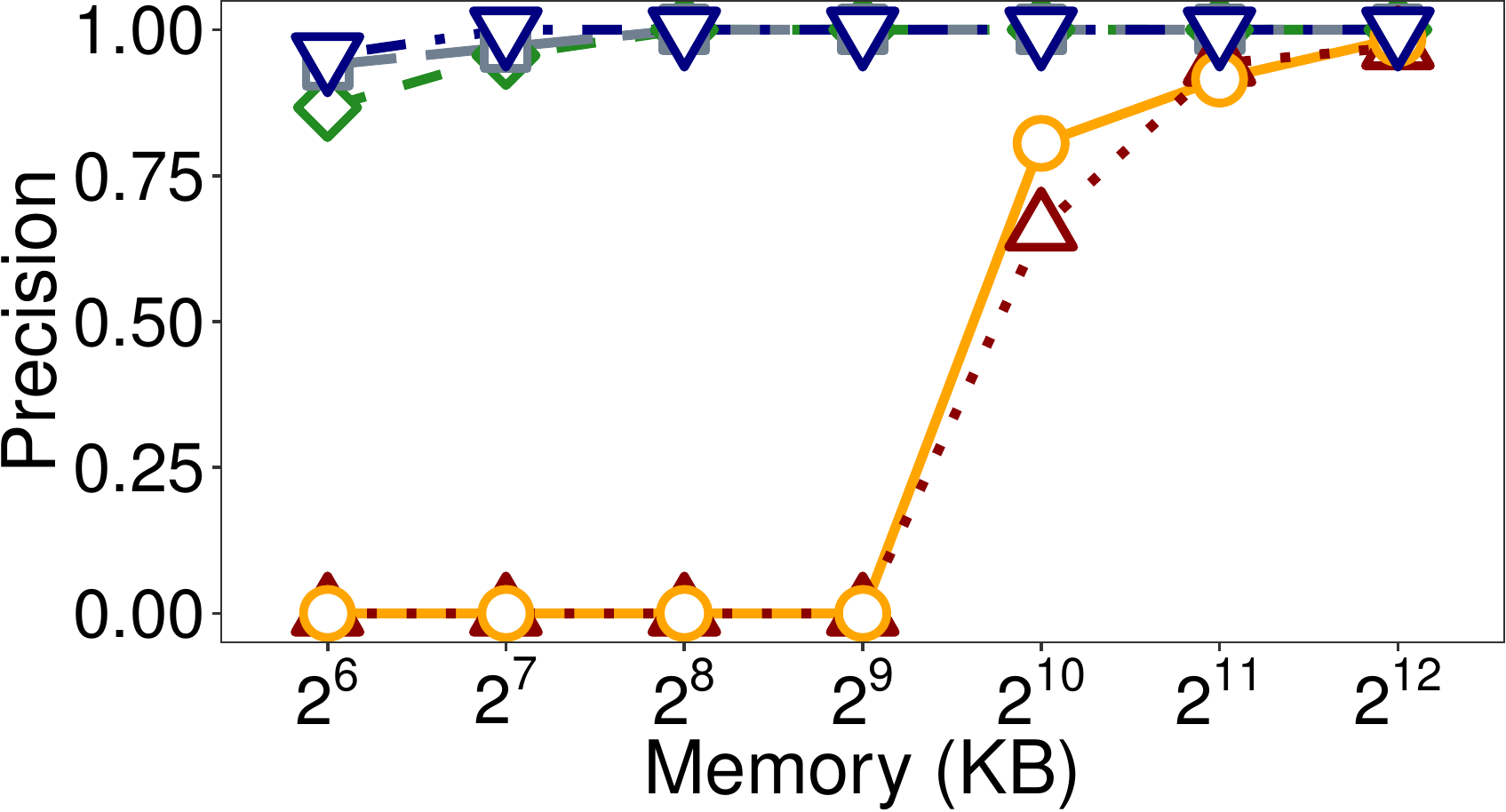} &
\includegraphics[width=1.65in]{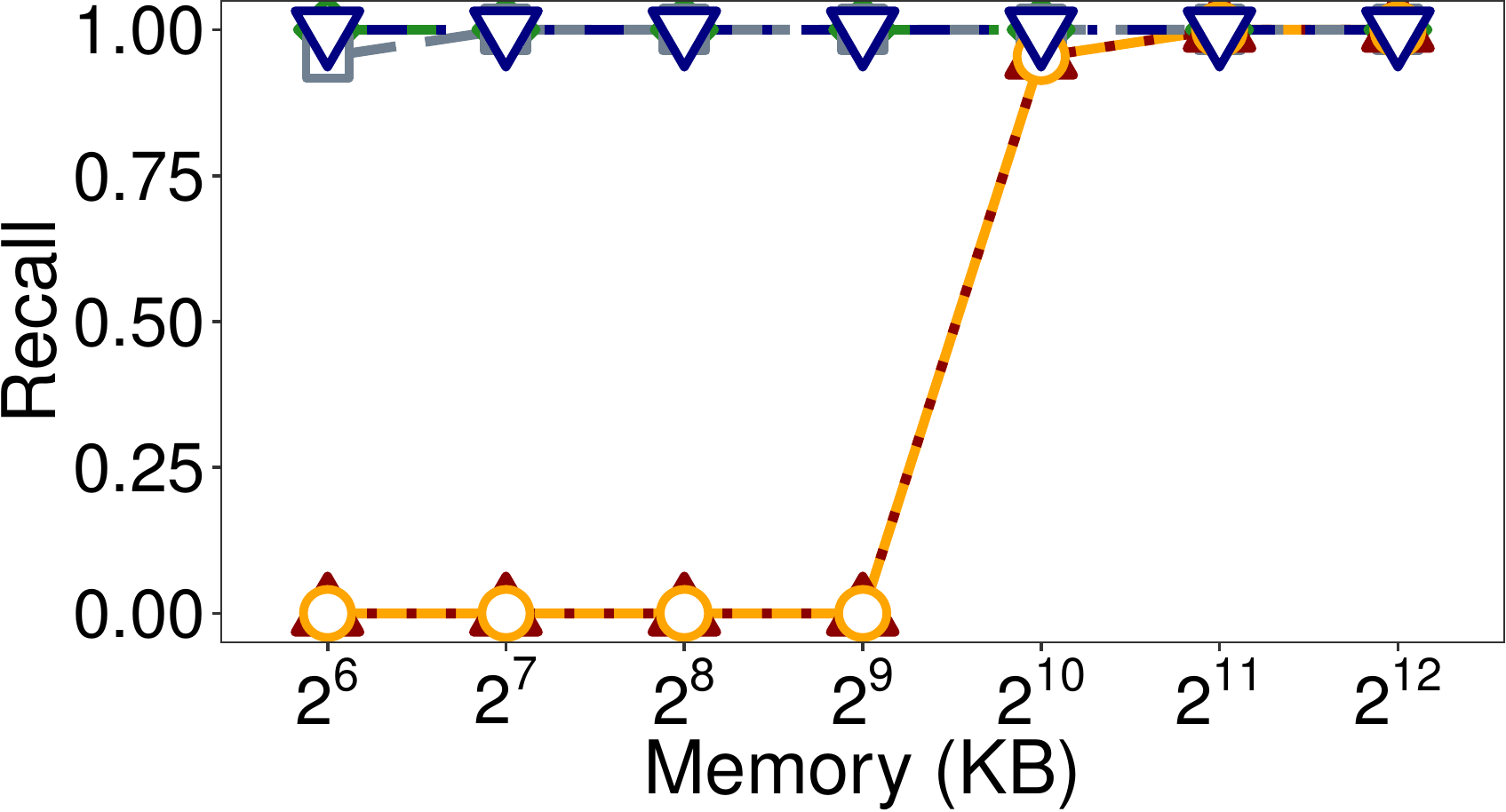} 
\vspace{-3pt}\\
{\footnotesize (a) Heavy hitter precision} & 
{\footnotesize (b) Heavy hitter recall} 
\vspace{3pt}\\
\includegraphics[width=1.65in]{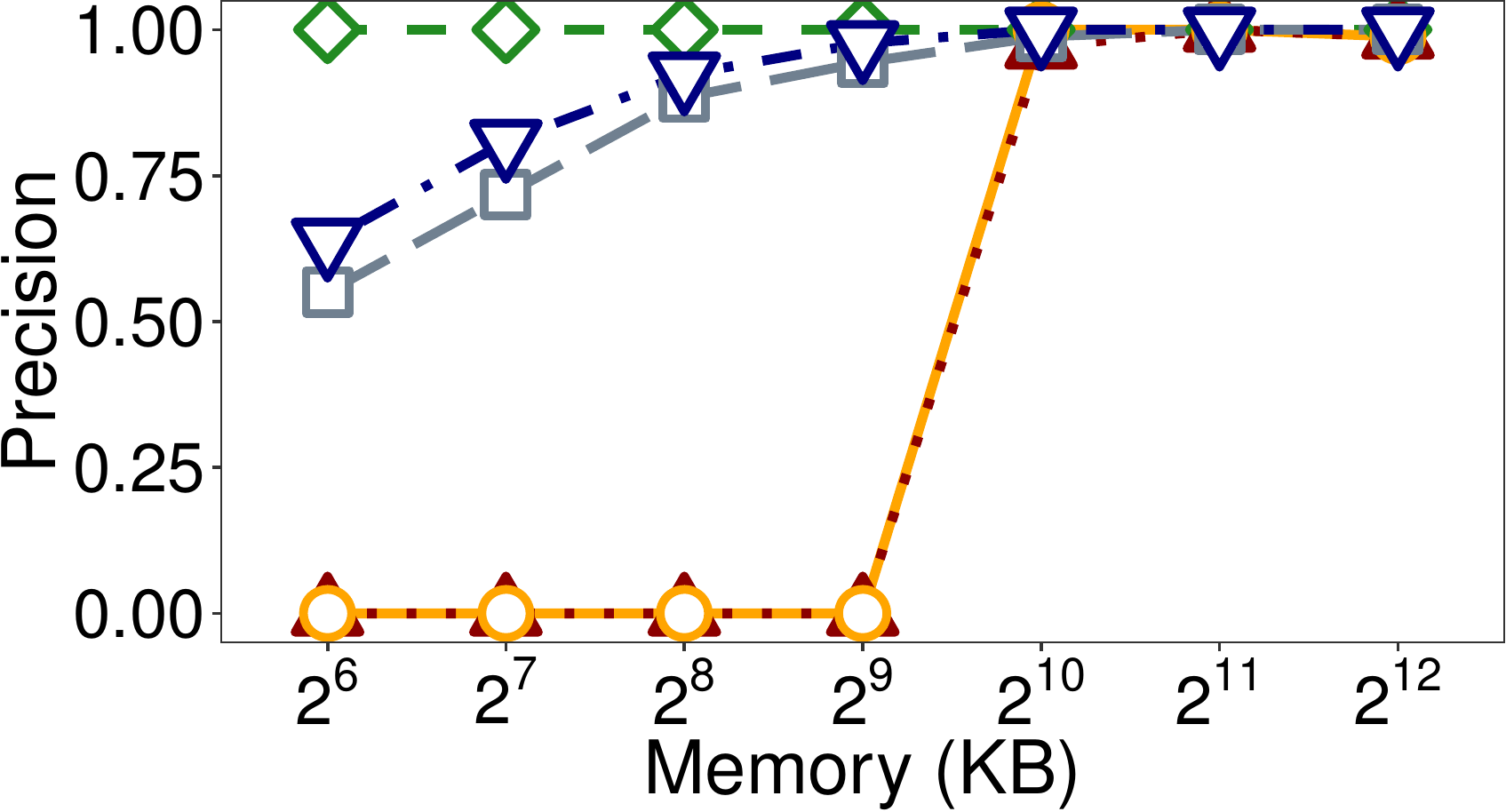} & 
\includegraphics[width=1.65in]{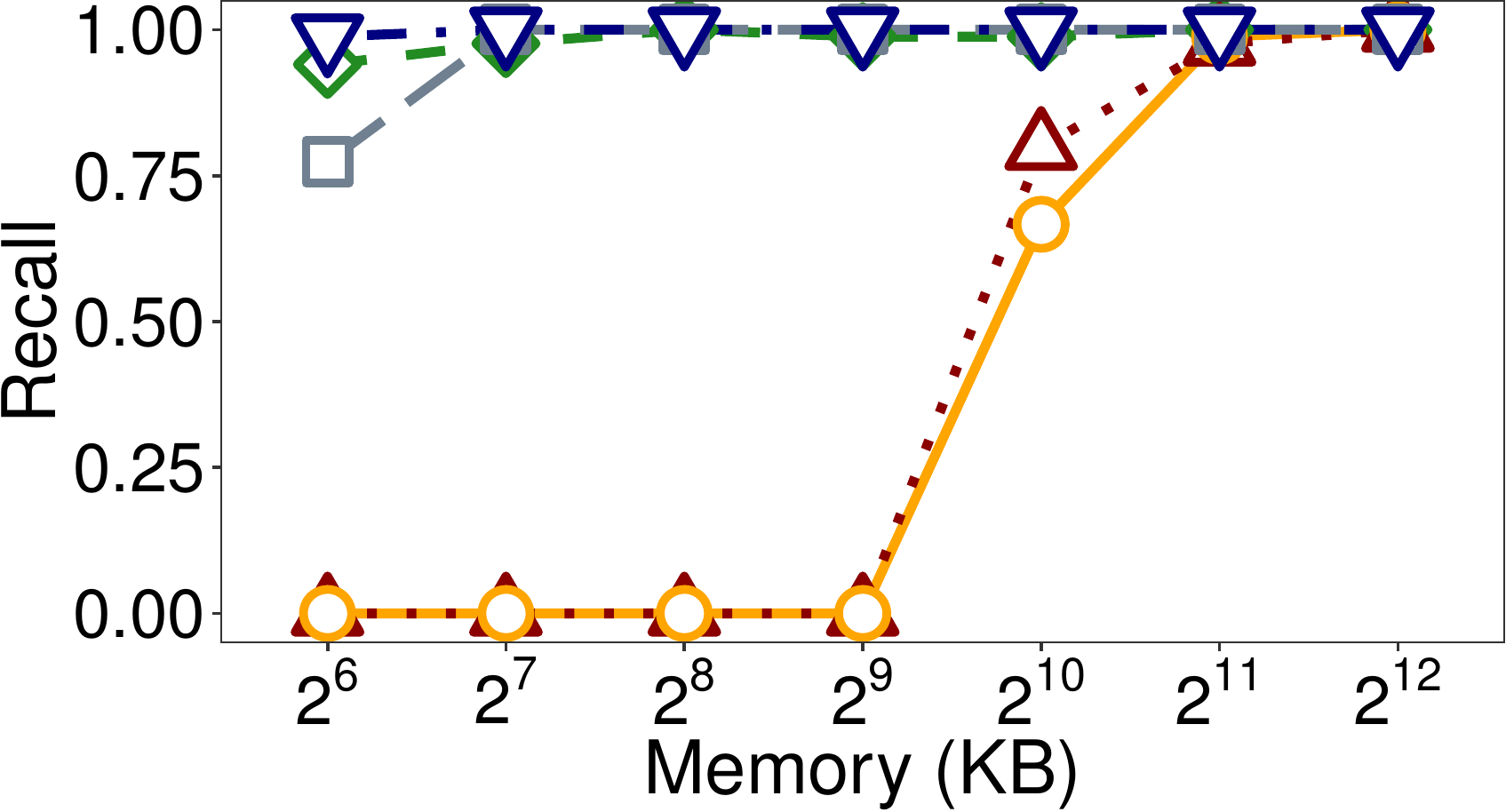} 
\vspace{-3pt}\\
{\footnotesize (c) Heavy changer precision} & 
{\footnotesize (d) Heavy changer recall}
\end{tabular}
\vspace{-6pt}
\caption{Experiment 4 (Accuracy for scalable detection).}
\label{fig:exp_dis}
\end{figure}

\para{Experiment 4 (Accuracy for scalable detection).}
Figure~\ref{fig:exp_dis} shows the precision and recall for scalable heavy
flow detection, in which we set $d=3$ and $q=5$.  We observe similar results
as in Experiments~1 and 2.  Note that we also conduct experiments with
different combinations of different settings of $d$ and $q$, and the results
show similar trends. 

\begin{figure}[!t]
\centering
\begin{tabular}{c@{\ }c}
\multicolumn{2}{c}{\includegraphics[width=1.8in]{fig/legend.pdf}} \\
\includegraphics[width=1.65in]{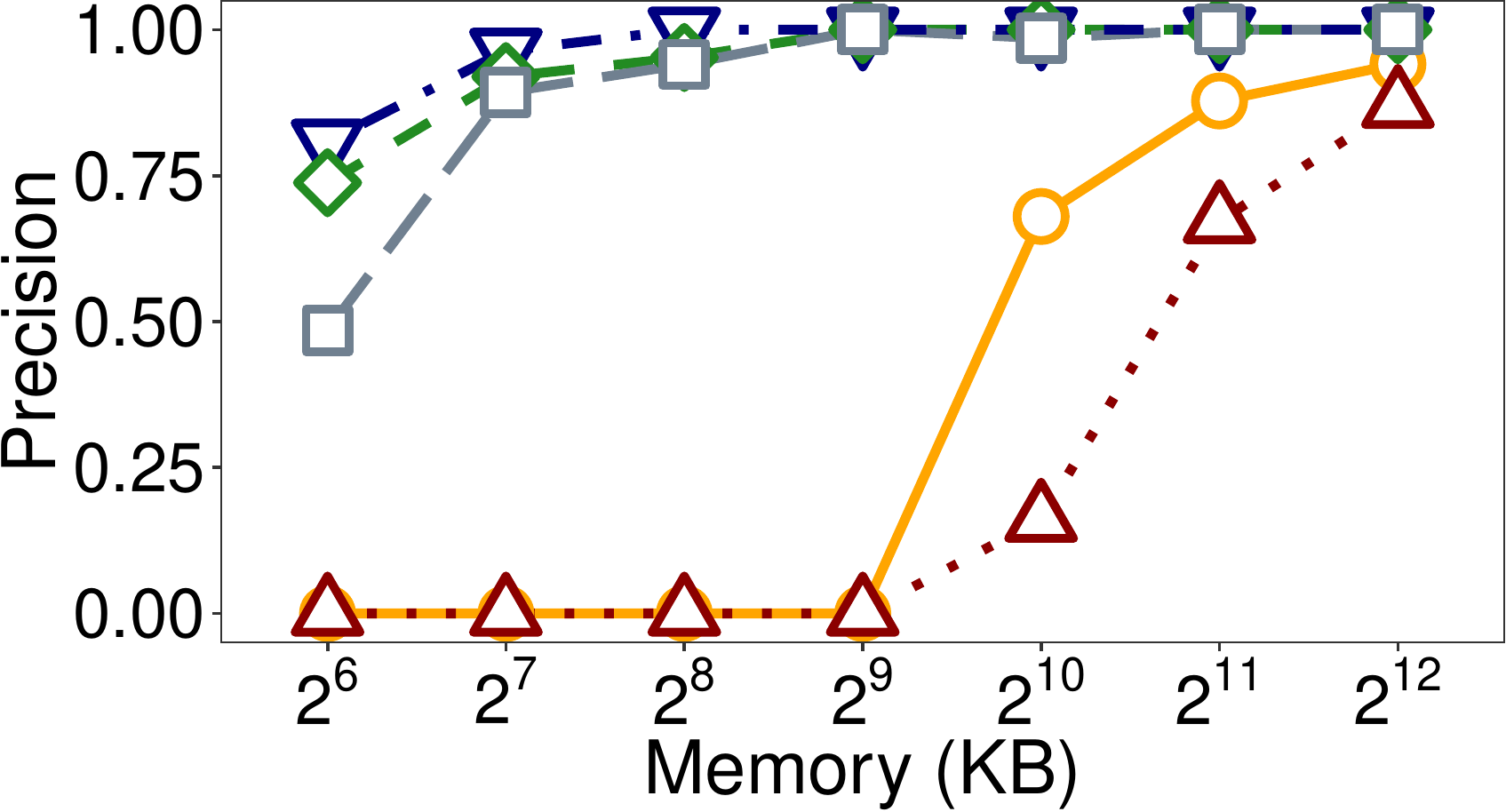} &
\includegraphics[width=1.65in]{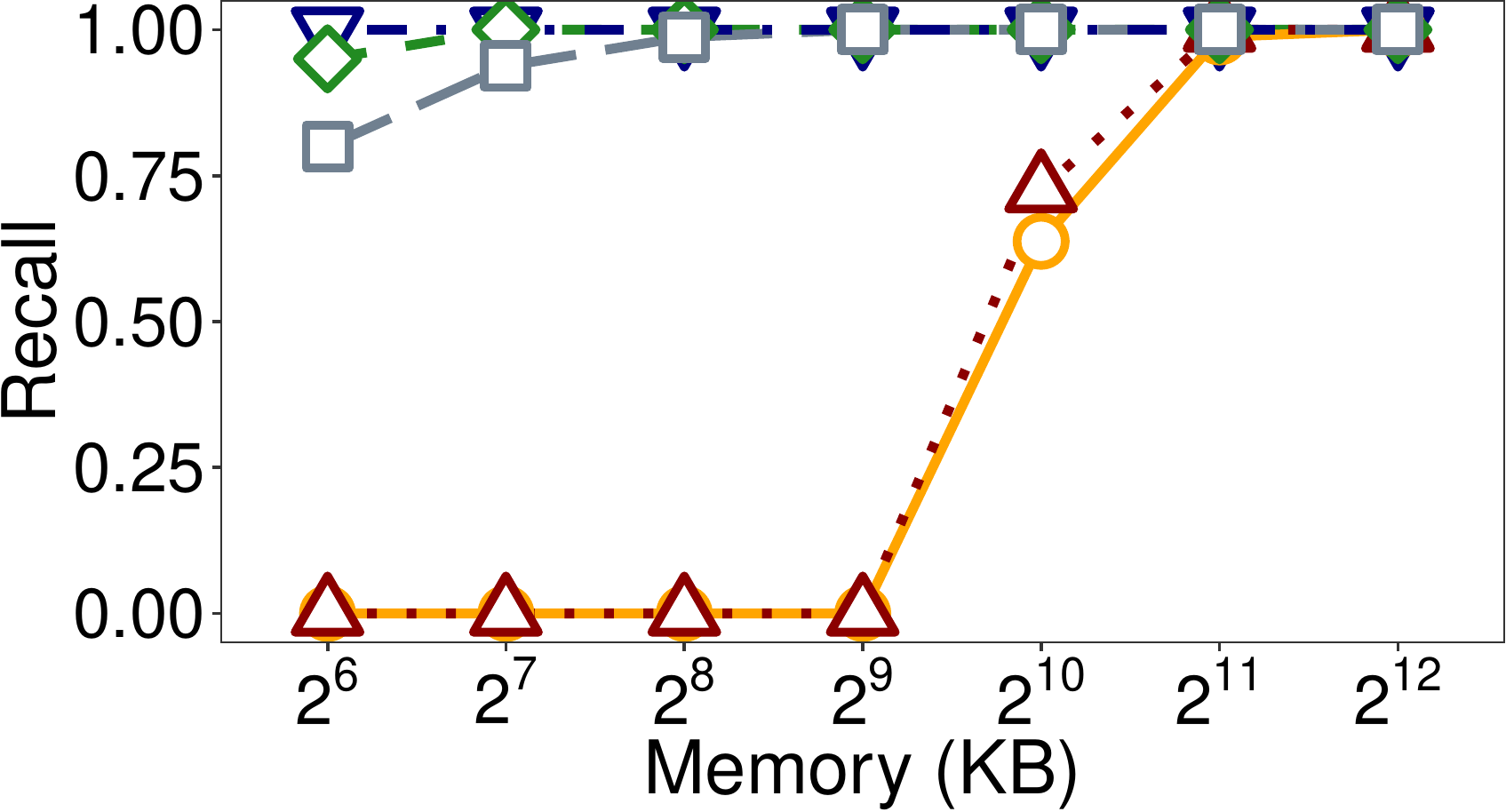} 
\vspace{-3pt}\\
{\footnotesize (a) Heavy hitter precision} & 
{\footnotesize (b) Heavy hitter recall} 
\vspace{3pt}\\
\includegraphics[width=1.65in]{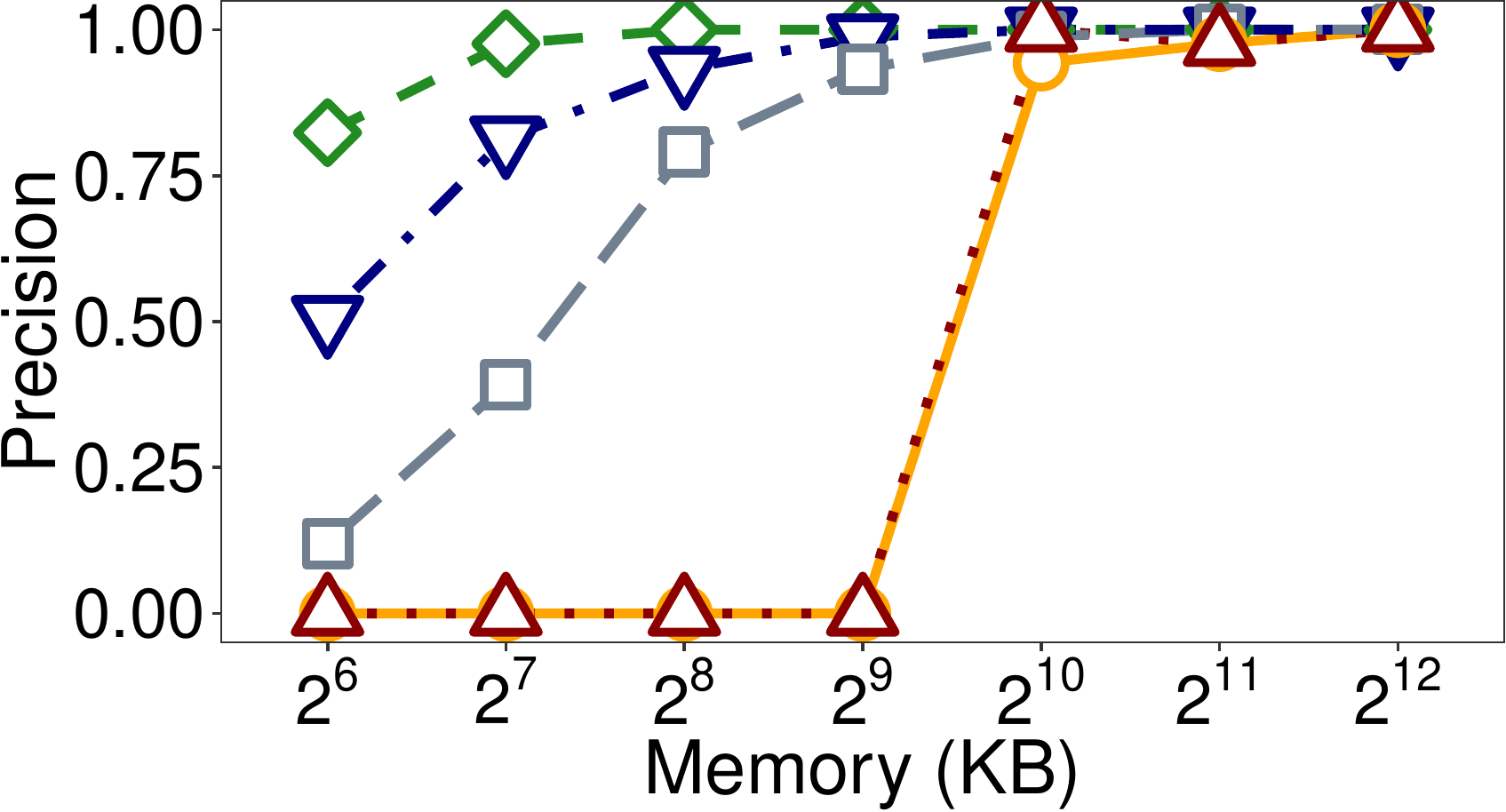} & 
\includegraphics[width=1.65in]{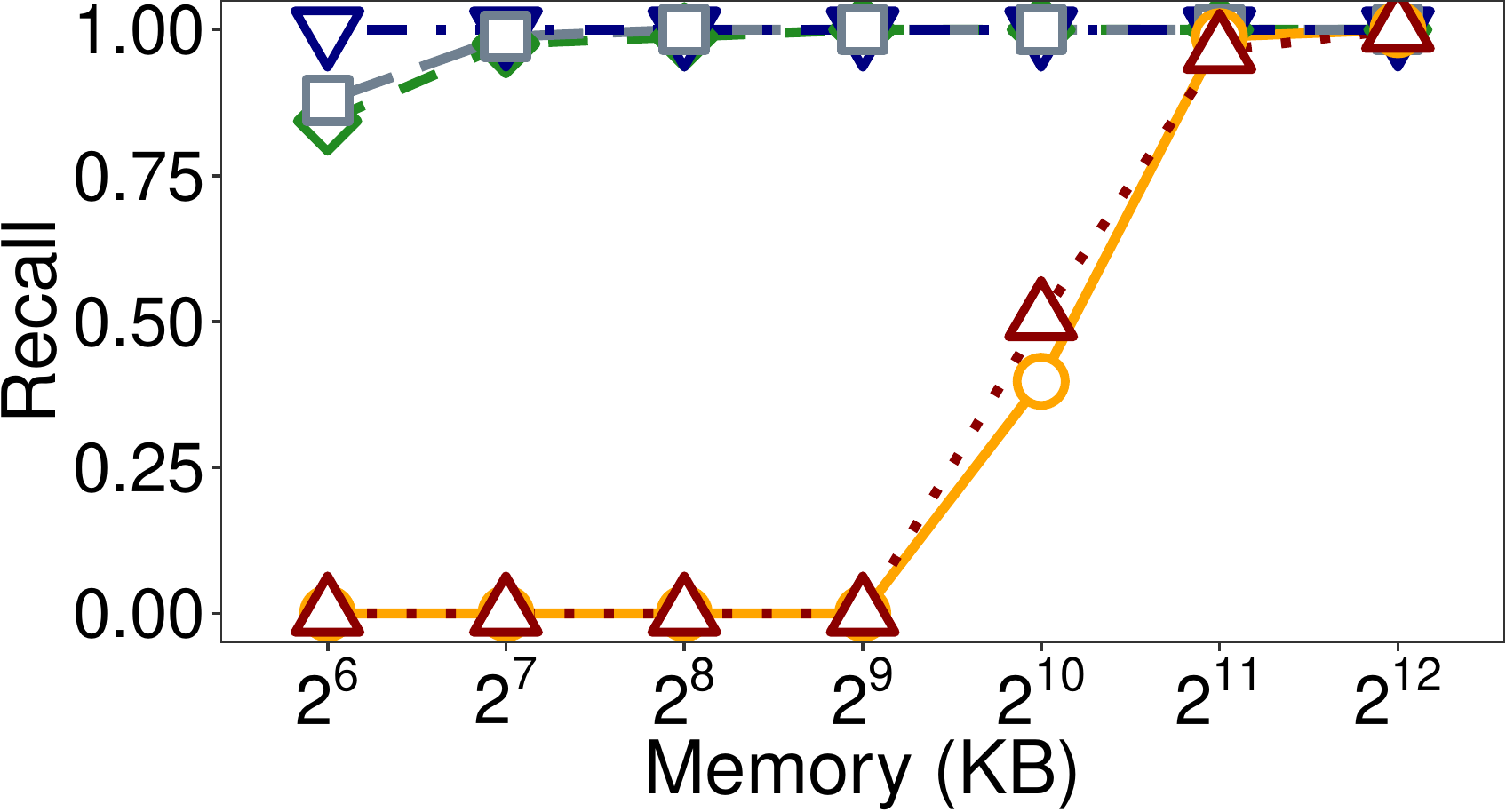} 
\vspace{-3pt}\\
{\footnotesize (c) Heavy changer precision} & 
{\footnotesize (d) Heavy changer recall}
\end{tabular}
\vspace{-6pt}
\caption{Experiment 5 (Accuracy for network-wide detection).}
\label{fig:exp_nw}
\vspace{-9pt}
\end{figure}

\para{Experiment 5 (Accuracy for network-wide detection).} We compare the
accuracy of all sketches in network-wide detection by varying the memory usage
in detectors. We randomly partition the 5-minute trace to six detectors, such
that the traffic of a flow is distributed in any non-empty subset of the six
detectors. Deltoid and Fast Sketch support network-wide detection inherently
due to their linear property, in which the counters of different sketch
instances with the same index can be added together. For Count-Min-Heap and
LD-Sketch, we obtain the estimated sum of each tracked flow key in every
detector and aggregate the estimated sums of each flow key.  We then use the
aggregates for heavy flow detection. 

Figure~\ref{fig:exp_nw} shows the results. Again, we observe similar results as
in Experiments~1 and 2. Note that the recall of \sysname is one in all memory
sizes for both the heavy hitter and heavy changer detection.  

\para{Experiment 6 (Performance optimizations of \sysname).}  We make a
case that \sysname can use SIMD instructions to process multiple data units in
parallel and achieve further performance gains.  Such performance
optimizations enable \sysname to address the need of fast network measurement
in software packet processing \cite{Liu2016,Li2016flowradar,Huang2017}. 

Here, we optimize the performance of the update operation
(Algorithm~\ref{alg:update}).  Specifically, we divide a hash value into $r$
parts (where $r=4$ in our case).  We use SIMD instructions to compute the
bucket indices of all $r$ rows, load the $r$ candidate heavy flow keys to a
register array, and compare the flow key with the $r$ candidate heavy flow
keys in parallel.  Based on the comparison results, we update the buckets 
(Algorithm~\ref{alg:update}).  For 64-bit keys, we use the AVX2 instruction
set to manipulate 256 bits (i.e., four 64-bit keys) in parallel.

Figure~\ref{fig:exp_simd} compares the original and optimized implementations
of \sysname.  The optimized version achieves 75\% higher throughput than the
original version on average.  Its throughput is above 14.88M~pkts/s in
most cases, implying that it can match the 10\,Gb/s line rate 
(Section~\ref{sec:introduction}). 

\begin{figure}[!t]
\centering
\includegraphics[width=2in]{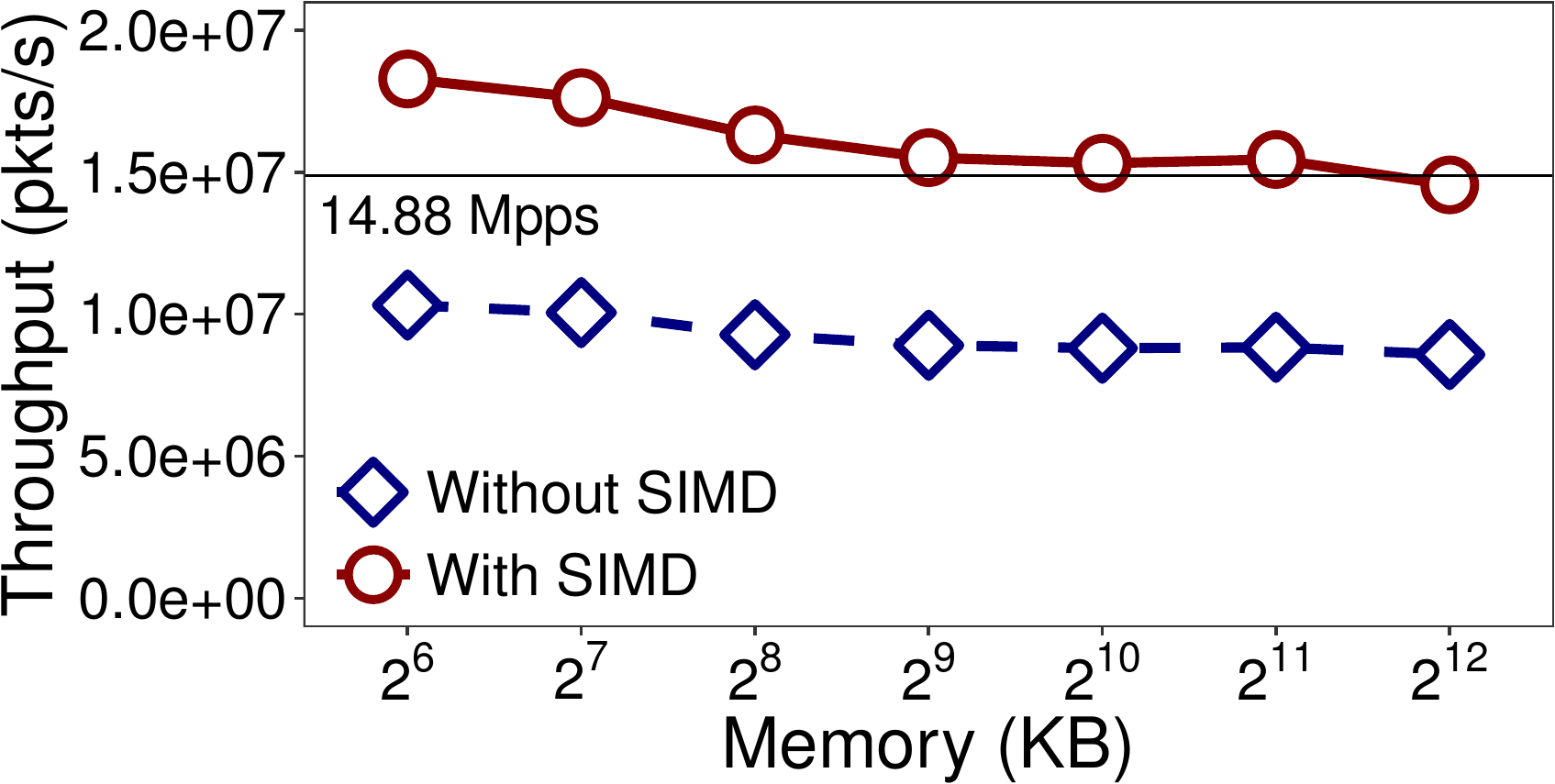} 
\vspace{-9pt}
\caption{Experiment 6 (Performance optimizations of \sysname).}
\label{fig:exp_simd}
\vspace{-9pt}
\end{figure}

\subsection{Evaluation in Hardware}
\label{subsec:evaluation_hw} 

We prototype \sysname in P4 \cite{Bosshart2014} and compile it to the Barefoot
Tofino chipset \cite{tofino}. 

\para{Testbed.} Our testbed consists of two servers and a Barefoot Tofino
switch. Each server has two 12-core 2.2\,GHz CPUs, 32\,GB RAM, and a 40\,Gbps
NIC, while the switch has 32 100\,Gb ports.  The two servers are connected via
the switch, where the traffic from one server is directly forwarded to the
other via the switch. 

\para{Methodology.}  We compare \sysname with PRECISION
\cite{Ben2018precision}, which is designed for heavy hitter
detection in programmable switches. PRECISION tracks heavy hitters by
probabilistically recirculating a small fraction of packets.  We compare
\sysname and PRECISION for both packet counting and size counting. We use the
same CAIDA trace as in Section~\ref{subsec:evaluation_sw}.
%with some probability that is inversely proportional to the minimal sum of
%the flow kept in it. 

We fix \sysname as $r=1$ row and 2,048 buckets. Our software evaluation
(Section~\ref{subsec:evaluation_sw}) shows that with such a configuration,
\sysname achieves an accuracy of above 0.9 for various epoch lengths.  We
configure PRECISION with 2-way associativity to balance between the accuracy
and the number of pipelined stages, and fix its memory usage to be the same
as that of \sysname. By default, we use source IPv4 addresses as flow keys.  

We use the built-in hash function CRC32 \cite{crc32} in our switch as the hash
function in sketches.  We find that both MurmurHash \cite{murmurhash} (used in
our software evaluation) and CRC32 have nearly identical accuracy results in
our evaluation, yet MurmurHash has much higher update throughput than CRC32.
Thus, we use MurmurHash in software evaluation, while using CRC32 here in
hardware evaluation.

\para{Experiment 7 (Switch resource usage).}  We measure the switch
resource usage of \sysname and PRECISION. We prototype PRECISION and each
version of \sysname (i.e., Algorithm~\ref{alg:mvgen} for 5-tuple flow
keys (MVFULL) and Algorithm~\ref{alg:mvsc} for size counting on 32-bit flow
keys (MVSC), and Algorithm~\ref{alg:mvpc} for packet counting on 32-bit flow
keys (MVPC)).  Note that PRECISION considers only 32-bit flow keys.

\begin{table}[!t]
\centering
\caption{Experiment~7: Switch resource usage (percentages in brackets are
fractions of total resource usage).}
\label{tab:tofino}
\renewcommand{\arraystretch}{1.05}
\vspace{-1em}
\scriptsize
\begin{tabular}{|c|c|c|c|c|c|}
\hline
    & MVFULL & MVSC & MVPC & PRECISION\\
\hline
\hline
    SRAM (KiB) &  144 (0.94\%) & 80 (0.52\%) &  80 (0.52\%) & 192 (1.25\%) \\ 
\hline 
No. stages &  4 (33.33\%) & 2 (16.67\%) & 1 (1.33\%) & 8 (66.67\%) \\ 
\hline 
No. actions & 10   & 5 &  3 & 15 \\
\hline 
No. ALUs & 3 (6.25\%) & 2 (4.17\%) & 2 (4.17\%) & 6 (12.5\%) \\    
\hline 
PHV (bytes) & 133 (17.32\%) & 108 (14.06\%) & 102 (13.28\%) & 137 (17.84\%) \\ 
\hline
\end{tabular}
\end{table}

Table~\ref{tab:tofino} shows the switch resource usage in terms of SRAM usage
(which measures memory usage), the numbers of physical stages, actions, and
stateful ALUs (all of which measure computational resources), as well as the
PHV size (which measures the message size across stages).  All the \sysname
implementations achieve less resource usage than PRECISION. 

\begin{figure}[!t]
\centering
\includegraphics[width=2in]{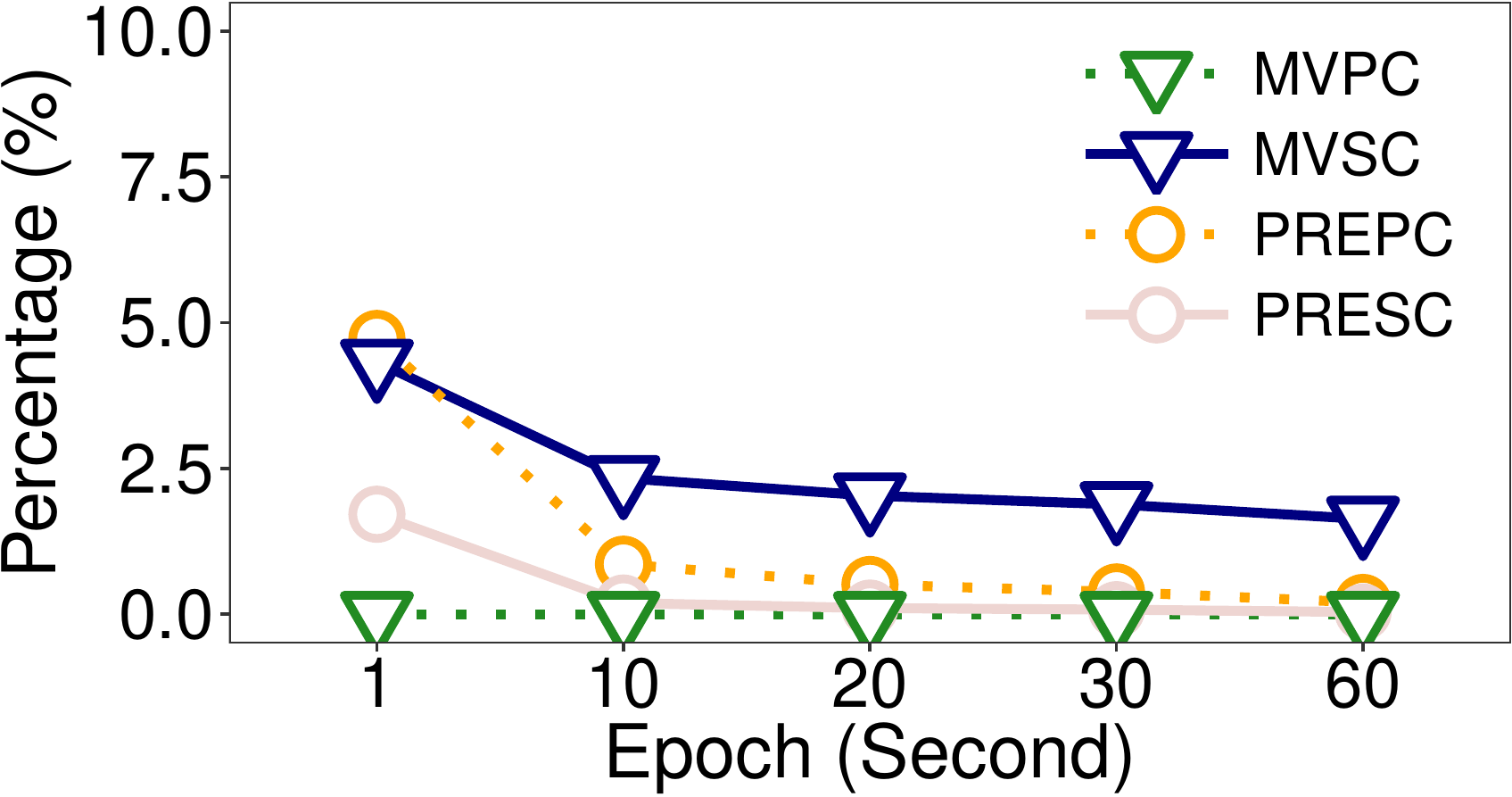} 
\vspace{-6pt}
\caption{Experiment 8 (Switch throughput analysis).  Here, we show the
percentage of packets being recirculated to the second pass of the switch
pipeline.}
\label{fig:exp_tofino_ratio}
\vspace{-12pt}
\end{figure}

\para{Experiment 8 (Switch throughput analysis).} We study the throughput
of \sysname and PRECISION. We split the trace into one-second epochs and
randomly select 50 epochs. We replay each epoch in one server with Pktgen-DPDK
\cite{pktgen} and compute the average receiving rate in another server with the
DPDK programs. 
We consider two implementation approaches of recirculating packets within the
switch pipeline: (i) using the {\em recirculate} primitive to recirculate
packets via the dedicated recirculation ports of the switch, and (ii) using
the {\em resubmit} primitive to recirculate packets within the ingress
pipeline via normal ports.  For the first approach, both \sysname and
PRECISION achieve the line rate in all cases.  For the second approach,
\sysname achieves the line rate for packet counting and 95\% of the line rate
in size counting, while PRECISION achieves around 98\% of the line rate for
both packet counting and size counting (not shown in figures).

To further examine the recirculation costs of \sysname and PRECISION in size
counting (denoted by MVSC and PRESC, respectively) and packet counting
(denoted by MVPC and PREPC, respectively), we conduct software simulation that
measures the
percentage of packets that are recirculated to the second pass of the switch
pipeline in an epoch for each approach.  Figure~\ref{fig:exp_tofino_ratio}
shows the results for different epoch lengths; since the variance of the
percentage for each approach is small, we omit the error bars here.  \sysname
has zero percentage for packet counting by design.  Also, it has a higher
percentage than PRECISION in size counting, yet the percentage is below 5\% in
all cases and shows a downward trend as the epoch length increases.  The
percentage of PRECISION in both packet counting and size counting is below
2\% in almost all cases.    

\para{Experiment 9 (Accuracy comparisons between \sysname and
PRECISION).} We study the accuracy of \sysname and
PRECISION for heavy hitter detection in software simulation by varying the
epoch lengths.  Figure~\ref{fig:exp_tofino_accuracy} shows the results.
\sysname achieves higher accuracy than PRECISION in size counting for all
epochs (in F1-score and relative errors), and has a comparable F1-score with
PRECISION in packet counting.  The relative error of \sysname is higher than
PRECISION in packet counting, yet the difference is small as the highest error
of \sysname is less than 1.6\% in our evaluation. PRECISION performs much
better in packet counting than in size counting. The reason is that the
counter value of each flow in size counting is much larger than that in packet
counting, and hence leads to a smaller recirculating probability and causes
PRECISION to miss more packets in size counting. 

\begin{figure}[!t]
\centering
\begin{tabular}{c@{\ }c}
\multicolumn{2}{c}{\includegraphics[width=2.5in]{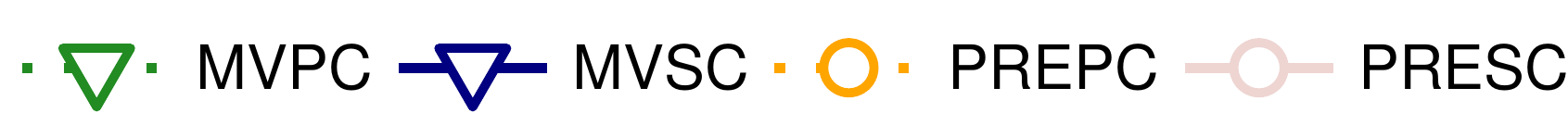}} \\
\includegraphics[width=1.6in]{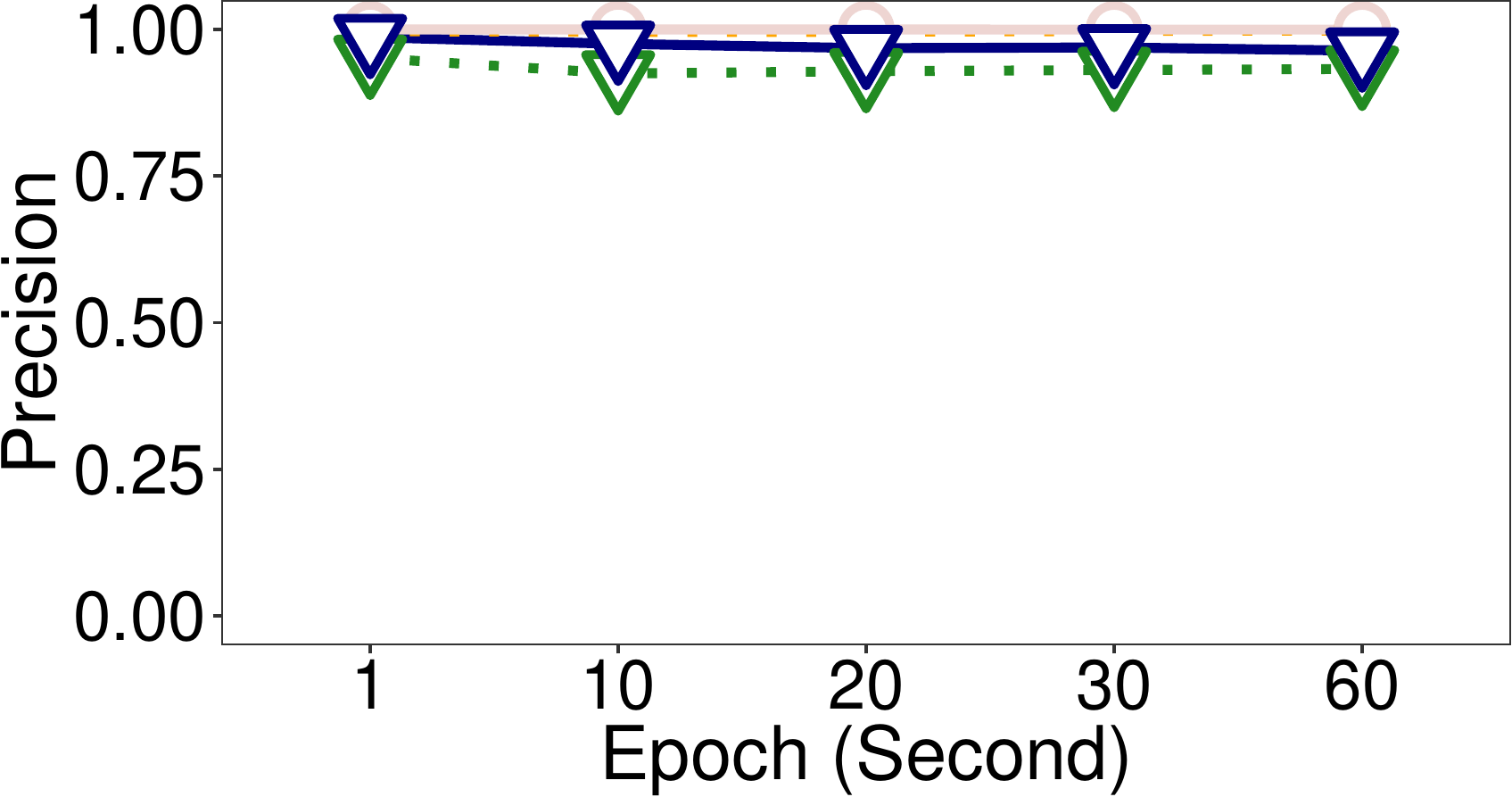} &
\includegraphics[width=1.6in]{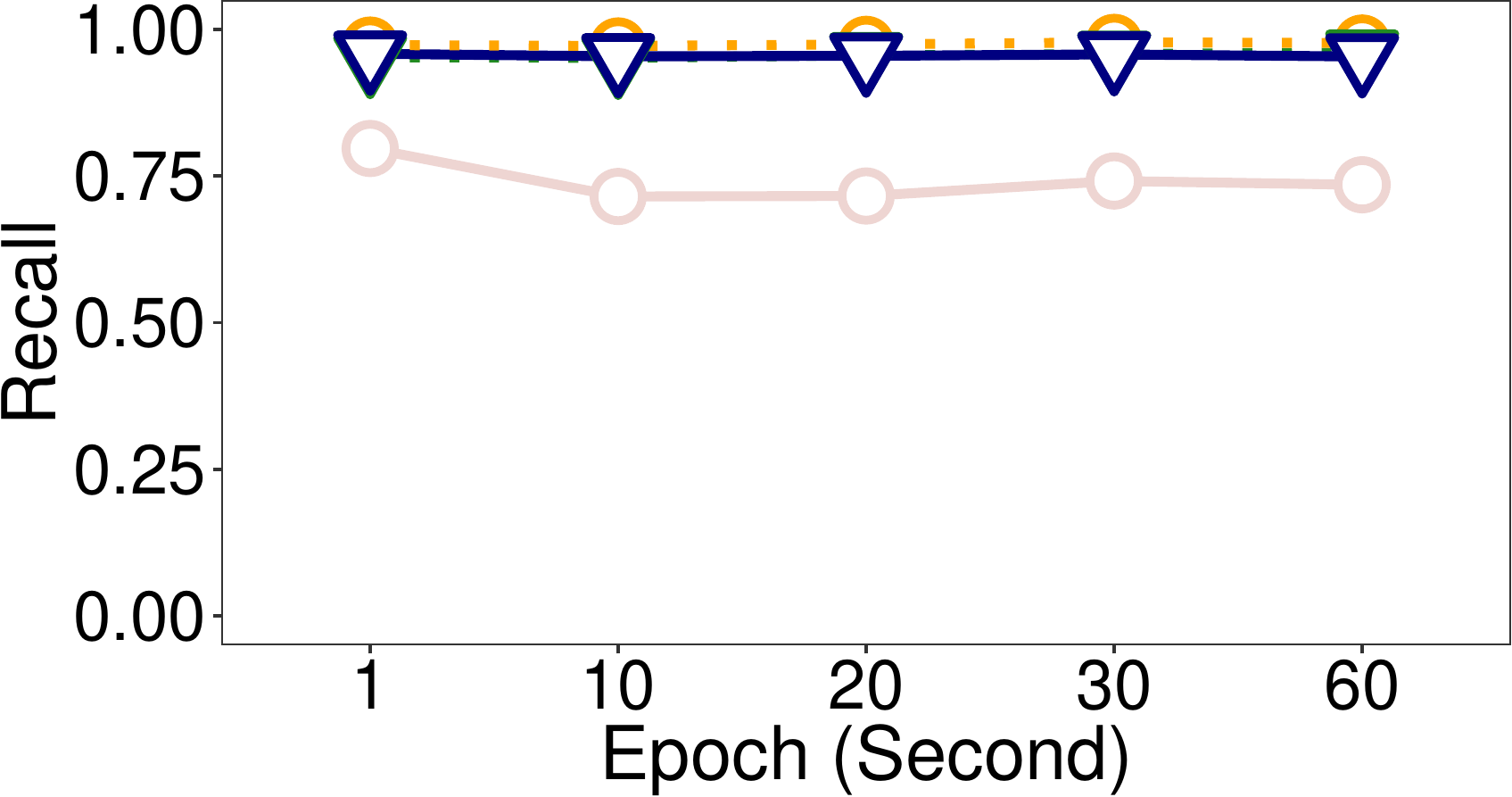} 
\vspace{-3pt}\\
{\footnotesize (a) Precision} & 
{\footnotesize (b) Recall}
\vspace{3pt}\\
\includegraphics[width=1.6in]{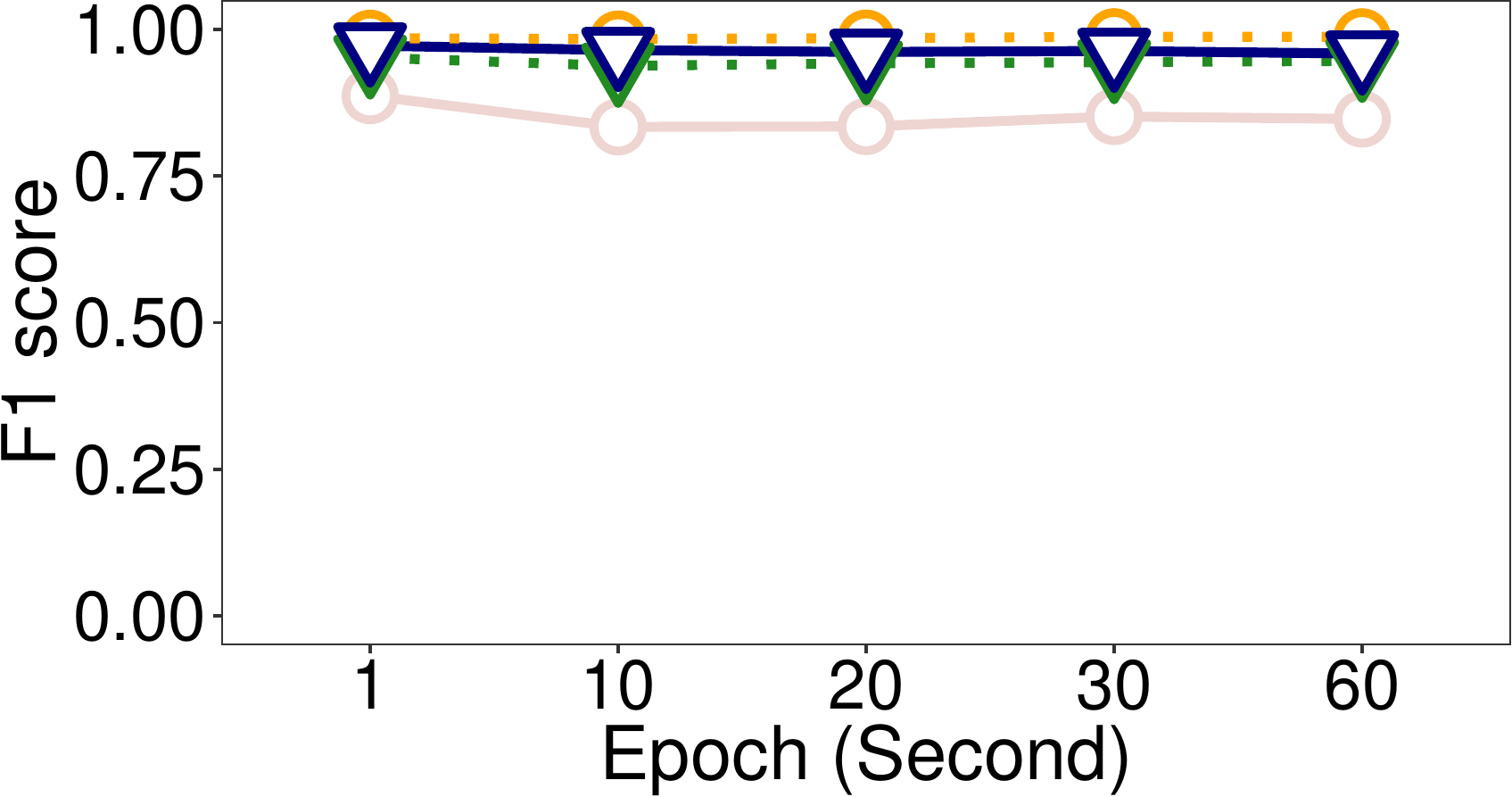} &
\includegraphics[width=1.6in]{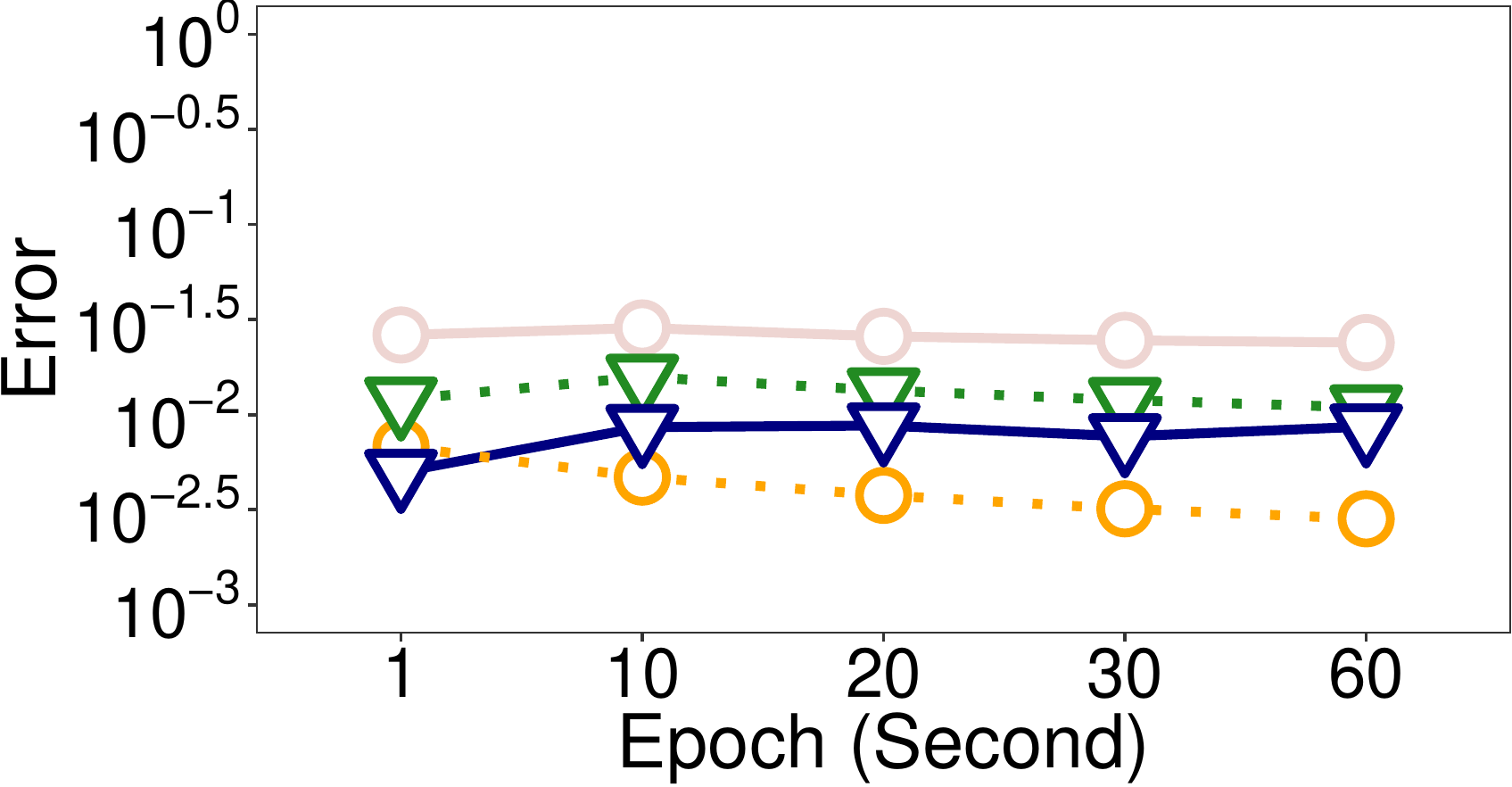} 
\vspace{-3pt}\\
{\footnotesize (c) F1 score} & 
{\footnotesize (d) Relative error} 
\end{tabular}
\vspace{-6pt}
\caption{Experiment 9 (Accuracy comparisons between \sysname and PRECISION). }
\label{fig:exp_tofino_accuracy}
\vspace{-9pt}
\end{figure}

%----------------------------------------------------------------------------
% Related Work
%----------------------------------------------------------------------------
\section{Related Work}
\label{sec:related} 

\para{Invertible sketches.}  In Section~\ref{subsec:sketches}, we review
several invertible sketches for heavy flow detection and their limitations.
Another related work extends the Bloom filter \cite{Bloom1970}
with invertibility \cite{Eppstein2011,Goodrich2015}.  In particular, the
Invertible Bloom Lookup Table (IBLT) \cite{Goodrich2015} tracks three
variables in each bucket: the number of keys, the sum of keys, and the sum of
values for all keys hashed to the bucket.  To recover all hashed keys, it
iteratively recovers from the buckets with only one hashed key and deletes the
hashed key of all its associated buckets (so that some buckets now have one
hashed key remaining).  FlowRadar \cite{Li2016flowradar} builds on IBLT for
heavy flow detection.  However, IBLT is sensitive to hash collisions: if
multiple keys are hashed to the same bucket, it fails to recover the keys in
the bucket.  

A closely related work to ours is AMON \cite{Kallitsis2016}, which applies
MJRTY in heavy hitter detection.  However, AMON and \sysname have different
designs: AMON splits a packet stream into multiple sub-streams and tracks the
candidate heavy flow for each sub-stream using MJRTY, while \sysname maps each
packet to the buckets in different rows in a sketch data structure.  \sysname
addresses the following issues that are not considered by AMON: (i) providing
theoretical guarantees on the trade-offs across memory usage, update/detection
performance, and detection accuracy; and (ii) addressing heavy changer
detection and network-wide detection. 

\para{Sketch-based network-wide measurement.} Recent studies
\cite{Yu2013,Moshref2015,Liu2016,Li2016flowradar,Huang2017,Yang2018,Huang2018}
propose sketch-based network-wide measurement systems for general measurement
tasks, including heavy flow detection.  Such systems leverage a centralized
control plane to analyze measurement results from multiple sketches in the
data plane.  Our work focuses on a compact invertible sketch design that
targets both heavy hitter and heavy changer detection. 

\para{Counter-based algorithms.}  Some approaches 
\cite{Misra1982,Karp2003,Sivaraman2017,Basat2017,Basat2017a, Yang2019}
track the most frequent flows in counter-based data structures (e.g., heaps
and associative arrays), which dynamically admit or evict flows based on
estimated flow sizes.  They target heavy hitter detection, but do not consider 
heavy changer detection and network-wide detection.

\para{Measurement in programmable switches.} Recent work focuses on
pushing measurement algorithms from end hosts to programmable switches 
\cite{Sivaraman2015,Jin2017,Sapio2017,Sivaraman2017,Harrison2018,Ben2018precision},
subject to the switch hardware constraints.  Our work demonstrates that
\sysname can be feasibly deployed in programmable switches to detect heavy
flows with limited resource overhead.

%---------------------------------------------------------------------
% Conclusion
%---------------------------------------------------------------------
\section{Conclusion}
\label{sec:conclusion} 

\sysname is an invertible sketch designed for fast and accurate heavy flow
detection.  It builds on the majority vote algorithm to enhance memory
management in two aspects: (i) small and static memory allocation, and (ii)
lightweight memory access in both update and detection operations.  It can
also be generalized for both scalable and network-wide detection.
Trace-driven evaluation in software demonstrates the throughput and accuracy
gains of \sysname.  We also show how the update performance of \sysname can be
boosted via SIMD instructions.  Evaluation in hardware demonstrates that
\sysname can be feasibly implemented in programmable switches with limited
resource overhead.

\bibliographystyle{abbrv}
\bibliography{paper}

%\vspace{-0.4in}
%\begin{IEEEbiographynophoto} {Lu Tang} is working toward the Ph.D. degree in
%Computer Science and Engineering at the Chinese University of Hong Kong. Her
%research interests are in network measurement. 
%\end{IEEEbiographynophoto}
%
%
%\vspace{-0.4in}
%\begin{IEEEbiographynophoto} {Qun Huang} received the Ph.D. degree in Computer
%Science and Engineering from the Chinese University of Hong Kong in 2015. He
%is now an Assistant Professor of the Department of Computer Science and
%Technology at Peking University.  His research interests are in distributed
%stream processing and network measurement. 
%\end{IEEEbiographynophoto}
%
%
%\vspace{-0.4in}
%\begin{IEEEbiographynophoto} {Patrick P. C. Lee} received the Ph.D. degree in
%Computer Science from Columbia University in 2008. He is now an Associate
%Professor of the Department of Computer Science and Engineering at the Chinese
%University of Hong Kong. His research interests are in various applied/systems
%topics including storage systems, distributed systems and networks, and cloud
%computing. 
%\end{IEEEbiographynophoto}

\end{document}